\documentclass[submission,copyright,creativecommons]{eptcs}
\providecommand{\event}{LSFA 2026} % Name of the event you are submitting to

\usepackage{iftex}
\usepackage{graphicx} % Required for inserting images
\usepackage{enumitem}
\usepackage{amsmath}
\usepackage{amscd}
\usepackage{amsfonts}
\usepackage{amsthm} 
\newcommand{\oc}{\mathord{!}}
\newcommand{\wn}{\mathord{?}}
\newcommand{\with}{\mathbin{\&}}
\usepackage{bussproofs}
\usepackage{amssymb}
\usepackage{mathtools}
\usepackage{latexsym}
\usepackage{color}
\usepackage{multicol}
\usepackage{url}
\newcommand\utimes{\mathbin{\ooalign{$\cup$\cr%
   \hfil\raise0.42ex\hbox{$\scriptscriptstyle\times$}\hfil\cr}}}
\usepackage{stmaryrd}
\newenvironment{bprooftree}
  {\leavevmode\hbox\bgroup}
  {\DisplayProof\egroup}
\makeatletter
\newtheorem{theorem}{Theorem}[section]
\newtheorem{lemma}[theorem]{Lemma}
\newtheorem{proposition}[theorem]{Proposition}
\newtheorem{corollary}[theorem]{Corollary}
\theoremstyle{definition}
\newtheorem{definition}[theorem]{Definition}

\theoremstyle{remark}
\newtheorem{remark}[theorem]{Remark}

\ifpdf
  \usepackage{underscore}         % Only needed if you use pdflatex.
  \usepackage[T1]{fontenc}        % Recommended with pdflatex
\else
  \usepackage{breakurl}           % Not needed if you use pdflatex only.
\fi

\title{Phase Semantic Cut-elimination for\\ Intuitionistic Linear Logic with Fixed Points}
\author{Jun Suzuki
\institute{Graduate School of Humanities and Human Sciences\\ Hokkaido University\\ Sapporo, Hokkaido, Japan}
\email{suzuki.jun.g2@elms.hokudai.ac.jp}
\and
Charles Grellois
\institute{School of Computer Science\\ University of Sheffield\\ Sheffield, United Kingdom}
\email{c.grellois@sheffield.ac.uk}
\and
Katsuhiko Sano
\institute{Faculty of Humanities and Human Sciences\\ Hokkaido University\\ Sapporo, Hokkaido, Japan}
\email{v-sano@let.hokudai.ac.jp}}
\def\titlerunning{Phase Semantic Cut-elimination for Intuitionistic Linear Logic with Fixed Points}
\def\authorrunning{J. Suzuki, C. Grellois \& K. Sano}
\begin{document}
\maketitle

\begin{abstract}
This paper establishes the cut-elimination theorem for intuitionistic propositional multiplicative-additive linear logic with the least and greatest fixpoints ($\mu \mathbf{IMALL}$) by means of its phase semantics. A classical first-order multiplicative-additive linear logic system with the least and greatest fixpoints was introduced by Baelde and Miller (2007). Its intuitionistic fragment was discussed in Baelde (2012), but the cut-elimination theorem for this fragment has not yet been proved. We introduce a propositional fragment of this system, $\mu \mathbf{IMALL}$, and establish the cut-elimination theorem. To prove the theorem, we define phase semantics for $\mu \mathbf{IMALL}$ and show the following two statements: (1) Soundness: if a formula is provable in $\mu \mathbf{IMALL}$, then it is true in all phase models, and (2) Cut-free Completeness: if a formula is true in all phase models, then it is provable in $\mu \mathbf{IMALL}$ without Cut. Okada (1999, 2002) employed a phase semantic method to prove the cut-elimination theorems for classical and intuitionistic linear logic systems. De et al. (2022) applied this method to a propositional fragment of classical propositional multiplicative-additive linear logic with the least and greatest fixpoints. We refine and apply their arguments to prove the cut-elimination theorem for $\mu \mathbf{IMALL}$.
\end{abstract}

\section{Introduction}\label{sec:Introduction}
This paper establishes the cut-elimination theorem for intuitionistic propositional multiplicative-additive linear logic with the least and greatest fixpoints ($\mu \mathbf{IMALL}$) in terms of its phase semantics.

A classical linear logic system with the least and greatest fixpoints was introduced by Baelde and Miller \cite{BaeldeMiller2007}. While standard classical linear logic introduced by Girard \cite{Girard1987} is a propositional sequent system that has exponential modalities $\oc$ and $\wn$ as well as multiplicative and additive connectives, Baelde and Miller's system is a first-order predicate (classical) linear logic with multiplicative-additive connectives, and has the least fixpoint operator $\mu$ and the greatest fixpoint operator $\nu$, instead of exponential modalities $\oc$ and $\wn$. By using these operators and their rules, we can enrich inference with induction and coinduction, and also simulate the exponentials via a translation we recall later in the paper. 

The \emph{intuitionistic} fragment of linear logic extended with the least and greatest fixpoints is discussed in Baelde~\cite[pp.7-8]{Baelde2012}. Since the intuitionistic system has a two-sided sequent calculus, the rules for $\mu$ and $\nu$ can be naturally derived from Knaster-Tarski's characterization of the least and greatest fixpoints in complete lattices (see Baelde \cite[pp.7-8]{Baelde2012}). The cut-elimination theorem for the intuitionistic fragment of linear logic extended with the least and greatest fixpoints, however, has not been proved. We introduce a propositional fragment of it, $\mu \mathbf{IMALL}$, and prove the cut-elimination for it. Cut-elimination for linear logic with the least and greatest fixpoints is essential in Baelde's framework since the system was introduced in the context of linear logic programming \cite{Andreoli1992}, where a search procedure for a cut-free proof is considered as a computation, as opposed to ordinary logic programming, where eliminating $Cut$s is regarded as a computation.

Furthermore, we have our own motivations for the intuitionistic fragment and cut-elimination for it. Intuitionistic proof systems are closely related to type-theoretic interpretations, and often give rise to computational interpretations via extensions of the Curry–Howard correspondence. In this perspective, the study of $\mu \mathbf{IMALL}$ and of its cut-elimination contributes to the development of type systems combining the resource-sensitive features of linear logic with inductive and coinductive definitions. Such systems provide a principled setting to reason about recursive and corecursive constructions, without imposing syntactic restrictions on their mutual nesting.

Beyond the modeling of inductive and coinductive data types, $\mu \mathbf{IMALL}$ also suggests a broader role as a specification logic. Recent work by Bauer and Saurin~\cite{bauer-saurin} shows that the modal $\mu$-calculus can be embedded into variants of classical linear logic with fixpoints. This suggests the possibility of using intuitionistic linear logics with fixpoints as type systems for programs with recursion, in which types express specification properties that programs need to satisfy. Such a framework seems adapted for a proof-theoretic study of higher-order model-checking (HOMC)
\cite{Ong06,KobayashiO09}
directly based on inductives and coinductives, and avoiding the detour via automata theory. This matters, as an important motivation for this line of research lies in its connection with proof assistants supporting inductive and coinductive reasoning, such as Coq, Agda, Lean, or Isabelle. By providing a proof-theoretic account of recursion and corecursion in a linear setting, $\mu \mathbf{IMALL}$ offers a promising framework to bridge theoretical developments in linear logic with program semantics and certified verification. Since HOMC is decidable, this could lead (if the translation has nice properties) to the identification of a decidable fragment of $\mu\mathbf{IMALL}$. This paper can be seen as a first step towards this goal: it establishes a crucial sanity check for $\mu\mathbf{IMALL}$, cut-elimination. One hope of our approach is to ultimately deliver a certified proof of the decidability of higher-order model-checking -- a first step being to provide a new proof, based on mathematical objects and concepts that are naturally close to proof assistants, since the mathematical complexity of existing proofs would lead to a very challenging certification process.

When it comes to syntactic cut-elimination, the rules for $\mu$ and $\nu$ make it difficult to prove the cut-elimination theorem because formulas in premises of the right rule and the left rule of $\mu$ (or $\nu$) may not match, and so a cut-elimination procedure may fail. For example, the rules for $\mu$ are of the following form in a propositional linear logic with the least and greatest fixpoints:

\[
    \begin{bprooftree}
        \AxiomC{$\Gamma \vdash \Delta, A(\mu x.A/x)$}
        \RightLabel{$({\vdash}\mu)$}
        \UnaryInfC{$\Gamma \vdash \Delta, \mu x.A$}
    \end{bprooftree},
    \begin{bprooftree}
        \AxiomC{$A(S/x) \vdash S$}
        \AxiomC{$\Gamma, S \vdash \Delta$}
        \RightLabel{$(\mu{\vdash})$}
        \BinaryInfC{$\mu x.A, \Gamma \vdash \Delta,$}
    \end{bprooftree}
\]
where $A(F/x)$ is a formula obtained by replacing all the occurrences of $x$ in $A$ with $F$. Consider the following proof with $Cut$
\[
    \begin{bprooftree}
        \AxiomC{$\Gamma_1 \vdash \Delta_1, A(\mu x.A/x)$}
        \RightLabel{$({\vdash}\mu)$}
        \UnaryInfC{$\Gamma_1 \vdash \Delta_1, \mu x.A$}
        \AxiomC{$A(S/x) \vdash S$}
        \AxiomC{$S, \Gamma_2 \vdash \Delta_2$}
        \RightLabel{$(\mu{\vdash})$}
        \BinaryInfC{$\mu x.A, \Gamma_2 \vdash \Delta_2$}
        \RightLabel{($Cut$)}
        \BinaryInfC{$\Gamma_1, \Gamma_2 \vdash \Delta_1, \Delta_2$}
    \end{bprooftree}
\]
Since $\mu x.A$ and $S$ may differ, we cannot reduce the $Cut$ to the $Cut$ between $A(\mu x.A/x)$ and $A(S/x)$ by standard arguments.

To solve this problem, three approaches are known. First, Baelde and Miller~\cite[Section 2]{BaeldeMiller2007ex} translate a formula with a fixpoint operator into a formula of second-order linear logic, and reduce the cut-elimination to that of focused second-order linear logic. Second, Baelde \cite[Section 3]{Baelde2012} proves the cut-elimination theorem by defining reduction rules and reducibility candidates. This method was also employed in Tiu \cite[Ch. 4]{Tiu2004} for a first-order intuitionistic system with induction and coinduction. Third, De et al. \cite{DeJafarrahmaniSaurin2022} employ phase semantics and reducibility candidates. In this paper, we adopt the third approach, a phase semantic approach. Phase semantics is a standard truth-value semantics introduced in Girard \cite[Section 1]{Girard1987}. De et al. \cite{DeJafarrahmaniSaurin2022} defined a phase semantics for the propositional fragment of multiplicative-additive linear logic with the least and greatest fixpoints, $\mu \mathbf{MALL}$\footnote{In Baelde-Miller \cite{BaeldeMiller2007} and Baelde \cite{Baelde2012}, $\mu \mathbf{MALL}$ denotes a first-order predicate multiplicative-additive linear logic. In this paper, however, we refer to propositional multiplicative-additive linear logic with the least and greatest fixpoints as $\mu \mathbf{MALL}$, in accordance with De et al. \cite{DeJafarrahmaniSaurin2022}.}, and proved the cut-elimination in terms of the semantics by showing the following two statements: (1) Soundness: if a formula is provable in $\mu \mathbf{IMALL}$, then it is true in all phase semantics, and (2) Cut-free Completeness: if a formula is true in all phase semantics, then it is provable in $\mu \mathbf{IMALL}$ without $Cut$. This method avoids considering a translation into second-order linear logic or reduction rules. Okada~\cite{Okada1999,Okada2002} proved the cut-elimination theorem for variations of linear logic including intuitionistic propositional linear logic ($\mathbf{ILL}$). We combine and refine Okada's and De et al.'s methods to prove the cut-elimination theorem for $\mu \mathbf{IMALL}$.

This paper is structured as follows. Section \ref{sec:Preliminaries} introduces propositional intuitionistic multiplicative-additive linear logic $\mathbf{IMALL}$, and provides an overview of Okada's semantic method \cite{Okada2002} for proving the cut-elimination theorem for $\mathbf{IMALL}$.
Section \ref{sec:SyntaxAndProofSystemForIMALLWithLeastAndGreatestFixedPoints} introduces syntax and sequent calculus system $\mu \mathbf{IMALL}$ for propositional intuitionistic multiplicative-additive linear logic with fixpoints, and establishes the functoriality lemma, a proof-theoretic property which plays an important role in our proof of the cut-free completeness. Section \ref{sec:PhaseSemanticsForMuIMALL} introduces a phase semantics for $\mu \mathbf{IMALL}$ and proves the soundness theorem. Section \ref{sec:CompletenessAndSemanticCutElimination} proves the cut-free completeness and the main theorem, the cut-elimination theorem. Section \ref{sec:ConclusionAndFutureWorks} concludes the paper with future research directions. 
\section{Preliminaries: IMALL and Phase Semantic Cut-elimination}\label{sec:Preliminaries}
In this section, we introduce the syntax and phase semantics for intuitionistic multiplicative-additive linear logic, $\mathbf{IMALL}$, and present Okada's semantic method \cite{Okada2002} for proving the cut-elimination theorem for $\mathbf{IMALL}$.

Syntax $\mathcal{L}$ of $\mathbf{IMALL}$ is defined as follows:
\[
    A \Coloneqq p \mid \mathbf{1} \mid \top \mid \mathbf{0} \mid A \otimes A \mid A \with A \mid A \oplus A \mid A \multimap A,
\]
where $p$ is an arbitrary element of the countably infinite set $\mathsf{Prop}$ of propositional variables. 
Greek letters $\Gamma, \Delta, \ldots$ denote finite \emph{multisets} of formulas. A \emph{sequent} is of the form $\Gamma \vdash C$, where the antecedent is a finite multiset of formulas and the succedent is exactly one formula. In the sequent, we write ``$\Gamma, \Delta$'' for $\Gamma \cup \Delta$.  
In the context of cut-free completeness later, we use ``$[A]$'' to denote the singleton multiset of $A$, while we simply write ``$A$'' within the sequent itself. 
A sequent with an empty antecedent, such as $\emptyset \vdash C$, is written as $\vdash C$. A \emph{sequent calculus} system of $\mathbf{IMALL}$ is shown in Table \ref{table:IMALL}. A \emph{proof} of $\mathbf{IMALL}$ is a finite tree generated by the $(\mathbf{id})$ axiom and the rules of $\mathbf{IMALL}$. We denote $\mathbf{IMALL}$ without $(Cut)$ by $\mathbf{IMALL}^-$.
\begin{table}[htbp]
    \caption{Sequent Calculus of $\mathbf{IMALL}$}
    \centering
    \renewcommand{\arraystretch}{2.3}
    \begin{tabular}{|cc|}
    \hline\begin{bprooftree}
  \AxiomC{}
  \RightLabel{(\textbf{id})}
  \UnaryInfC{$A \vdash A$}
  \end{bprooftree}
&
  \begin{bprooftree}
  \AxiomC{$\Gamma \vdash A$}
  \AxiomC{$A, \Delta \vdash B$}
  \RightLabel{($Cut$)}
  \BinaryInfC{$\Gamma, \Delta \vdash B$}
  \end{bprooftree} \\

  \multicolumn{2}{|c|}{
  \begin{bprooftree}
  \AxiomC{}
  \RightLabel{(${\vdash}\mathbf{1}$)}
  \UnaryInfC{$\vdash \mathbf{1}$}
  \end{bprooftree}
\begin{bprooftree}
  \AxiomC{$\Gamma \vdash C$}
  \RightLabel{($\mathbf{1}{\vdash}$)}
  \UnaryInfC{$\mathbf{1}, \Gamma \vdash C$}
\end{bprooftree}
\begin{bprooftree}
  \AxiomC{}
  \RightLabel{(${\vdash}\top$)}
  \UnaryInfC{$\Gamma \vdash \top$}
\end{bprooftree}
\begin{bprooftree}
  \AxiomC{}
  \RightLabel{($\mathbf{0}{\vdash}$)}
  \UnaryInfC{$\mathbf{0}, \Gamma \vdash C$}
\end{bprooftree}}
\\
  \begin{bprooftree}
  \AxiomC{$\Gamma \vdash A$}
  \AxiomC{$\Delta \vdash B$}
  \RightLabel{(${\vdash}\otimes$)}
  \BinaryInfC{$\Gamma, \Delta \vdash A \otimes B$}
  \end{bprooftree}
&
  \begin{bprooftree}
  \AxiomC{$A, B, \Gamma \vdash C$}
  \RightLabel{($\otimes{\vdash}$)}
  \UnaryInfC{$A \otimes B, \Gamma \vdash C$}
  \end{bprooftree}\\

  \begin{bprooftree}
  \AxiomC{$\Gamma \vdash A$}
  \AxiomC{$\Gamma \vdash B$}
  \RightLabel{(${\vdash}\with$)}
  \BinaryInfC{$\Gamma \vdash A \with B$}
  \end{bprooftree}
&
  \begin{bprooftree}
  \AxiomC{$A_i, \Gamma \vdash C$}
  \RightLabel{($\with{\vdash}$)}
  \UnaryInfC{$A_0 \with A_1, \Gamma \vdash C$}
  \end{bprooftree}\\

  \begin{bprooftree}
  \AxiomC{$\Gamma \vdash A_i$}
  \RightLabel{(${\vdash}\oplus$)}
  \UnaryInfC{$\Gamma \vdash A_0 \oplus A_1$}
  \end{bprooftree}
&
  \begin{bprooftree}
  \AxiomC{$A, \Gamma \vdash C$}
  \AxiomC{$B, \Gamma \vdash C$}
  \RightLabel{($\oplus{\vdash}$)}
  \BinaryInfC{$A \oplus B, \Gamma \vdash C$}
  \end{bprooftree}\\

  \begin{bprooftree}
  \AxiomC{$A, \Gamma \vdash B$}
  \RightLabel{($\vdash\multimap$)}
  \UnaryInfC{$\Gamma \vdash A \multimap B$}
  \end{bprooftree}
&
  \begin{bprooftree}
  \AxiomC{$\Gamma \vdash A$}
  \AxiomC{$B, \Delta \vdash C$}
  \RightLabel{($\multimap\vdash$)}
  \BinaryInfC{$A \multimap B, \Gamma, \Delta \vdash C$}
  \end{bprooftree}
\rule[-15pt]{0pt}{30pt}\\ \hline
    \end{tabular}
    \label{table:IMALL}
\end{table}
In the cut-free $\mathbf{IMALL}^-$, the following  inversion can be shown by induction on a proof.
\begin{lemma}[Inversion in $\mathbf{IMALL}^-$]\label{inversionForIMALL}
    \begin{enumerate}
        \item if $\Gamma \vdash A \multimap B$ is provable in $\mathbf{IMALL}^-$, then $A, \Gamma \vdash B$ is also provable in $\mathbf{IMALL}^-$.
        \item if $A \otimes B, \Gamma \vdash C$ is provable in $\mathbf{IMALL}^-$, then $A, B, \Gamma \vdash C$ is also provable in $\mathbf{IMALL}^-$.
    \end{enumerate}
\end{lemma}

Next, we introduce \emph{phase semantics} for $\mathbf{IMALL}$. Several formulations of phase semantics are known for intuitionistic linear logic (Abrusci \cite{Abrusci1990}, Troelstra \cite[Chapter 8]{Troelstra1992}), but here we adopt Okada's definition \cite[Section 2]{Okada2002}, which can be applied to a semantic argument for the cut-elimination theorem.

Let $M = (M, \cdot, 1)$ be a commutative monoid, where ``$\cdot$'' is a commutative and associative binary operation, and $1$ is the neutral element. For $X, Y \subseteq M$, we define that
\[
    X \cdot Y = \{x \cdot y \mid \text{$x \in X$ and $y \in Y$}\}.
\]
In what follows, we may omit the monoid operation ``$\cdot$'' and the parentheses to simply write, e.g.,  $x y$ and $XY$ to mean $x\cdot y$ and $X \cdot Y$, respectively, when no confusion arises.  
\begin{definition}\label{theSetOfClosedSets}
    A set $D_M \subseteq \wp(M)$ is  \emph{a set of closed sets} if it satisfies the following:
\begin{enumerate}
    \item for any $D' \subseteq D_M$, $\bigcap D' \in D_M$,
    \item for any $X \in \wp(M)$ and $Y \in D_M$, $X \multimap Y = \{y \in M \mid \text{for all $x \in X$, $x y \in Y$} \}\in D_M$.
\end{enumerate}
An \emph{intuitionistic phase space} is 
a pair $(M, D_M)$ of a commutative monoid and a set of closed sets. 
The \emph{closure} function $cl\colon\wp(M)\to\wp(M)$ is defined by $cl(X) = \bigcap\{Y\in D_M\mid X\subseteq Y\}$.
\end{definition}

The closure function has the following properties.
\begin{proposition}[{Okada \cite[p. 475]{Okada2002}}]\label{closureProposition}
Let $X, Y \subseteq M$. 
\begin{multicols}{2}
\begin{enumerate}
    \item $X \subseteq cl(X)$,
    \item $cl(cl(X)) = cl(X)$,
    \item if $X \subseteq Y$, then $cl(X) \subseteq cl(Y)$,
    \item $cl(X) \cdot Y \subseteq cl(X \cdot Y)$.
\end{enumerate}
\end{multicols}
\end{proposition}

The following explains why an element of $D_{M}$ is a \emph{closed set}:

\begin{proposition}
    For an intuitionistic phase space $(M, D_M)$, $D_M = \{cl(X) \mid X \subseteq M\}$ holds.
\end{proposition}

An \emph{intuitionistic phase model} is 
a triple $(M, D_M, V)$ where $(M, D_M)$ is an intuitionistic phase space and $V\colon \mathsf{Prop}\to D_M$ is a \emph{valuation function}. 
The \emph{interpretation} $\llbracket A \rrbracket^V$ of $A \in \mathcal{L}$ in an intuitionistic phase model $(M, D_M, V)$ is defined by induction as follows:
\begin{multicols}{2}
\begin{itemize}
    \item $\llbracket p \rrbracket^V = V(p)$,
    \item $\llbracket \mathbf{1} \rrbracket^V = cl(\{1\})$,
    \item $\llbracket \top \rrbracket^V = M$,
    \item $\llbracket \mathbf{0} \rrbracket^V = cl(\emptyset)$,
    \item $\llbracket A_0 \otimes A_1 \rrbracket^V = cl(\llbracket A_0 \rrbracket^V \cdot \llbracket A_1 \rrbracket^V)$,
    \item $\llbracket A_0 \with A_1 \rrbracket^V = \llbracket A_0 \rrbracket^V \cap \llbracket A_1 \rrbracket^V$,
    \item $\llbracket A_0 \oplus A_1 \rrbracket^V = cl(\llbracket A_0 \rrbracket^V \cup \llbracket A_1 \rrbracket^V)$,
\end{itemize}
\end{multicols}
\begin{itemize}
    \item $\llbracket A_0 \multimap A_1 \rrbracket^V = \llbracket A_0 \rrbracket^V \multimap \llbracket A_1 \rrbracket^V = \{b \in M \mid \text{for all $a \in \llbracket A_0 \rrbracket^V$, $ab \in \llbracket A_1 \rrbracket^V$}\}$.
\end{itemize}

The following can be shown easily by induction on $A$. 
\begin{proposition}
    Let $(M, D_M, V)$ be an intuitionistic phase model. Then, $\llbracket A \rrbracket^V \in D_M$ for any $A \in \mathcal{L}$.
\end{proposition}

A formula $A$ is \emph{true} in an intuitionistic phase model $(M, D_M, V)$ if $1 \in \llbracket A \rrbracket^V$. Then the following soundness holds. 
\begin{proposition}[{Okada \cite[Theorem 3.1]{Okada2002}}]\label{SoundnessForIMALL}
    If a sequent $A_1, \ldots, A_n \vdash C$ is provable in $\mathbf{IMALL}$, then $\llbracket A_1 \rrbracket^V \cdots \llbracket A_n \rrbracket^V \subseteq \llbracket C \rrbracket^V$ for any intuitionistic phase model $(M, D_M, V)$, where the zero-ary product of the monoid operation is $\{1\}$. In particular, if $\vdash A$ is provable in $\mathbf{IMALL}$, then $A$ is true in any intuitionistic phase model.
\end{proposition}

Below, we outline the proof of cut-free completeness of $\mathbf{IMALL}$ for intuitionistic phase models. The proof relies on the construction of a syntactic intuitionistic phase model. 

\begin{definition}
    For a formula $C$, we define 
    \[
    \mathsf{Pr}_{cf}(C) = \{\Gamma \mid \text{$\Gamma \vdash C$ is provable without $(Cut)$ in $\mathbf{IMALL}$}\}.
    \]
    Define the syntactic intuitionistic phase model $(M, D_M, V)$ as follows:
    \begin{itemize}
        \item The base set $M$ of the monoid is the set of all finite multisets of formulas.
        \item The monoid operation is the union of multisets $\cup$.
        \item The neutral element of the monoid is the empty multiset $\emptyset$.
        \item $X \in D_M$ iff $X = \bigcap\{\mathsf{Pr}_{cf}(C) \mid C \in \mathbb{F}\}$ for some set $\mathbb{F}$ of formulas.
        \item $V(p) := \mathsf{Pr}_{cf}(p)$.
    \end{itemize}
\end{definition}

This $D_M$ satisfies the two conditions of the set of closed sets in Definition \ref{theSetOfClosedSets}. The first condition is clearly satisfied. To show the second one, fix any $X \in \wp(M)$ and $Y = \bigcap \{\mathsf{Pr}_{cf}(C) \mid C \in \mathbb{F}\} \in D_M$. It suffices to show that $X \multimap Y = X \multimap \bigcap \{\mathsf{Pr}_{cf}(C) \mid C \in \mathbb{F}\} = \bigcap \{\mathsf{Pr}_{cf}(\bigotimes\Gamma \multimap C) \mid \Gamma \in X, C \in \mathbb{F}\}$, where $\bigotimes \emptyset = \mathbf{1}$ and $\bigotimes (\Gamma' \cup [A]) = (\bigotimes \Gamma') \otimes A$. This can be shown by Lemma \ref{inversionForIMALL} $($Inversion$)$.
    
The following lemma is a crucial step to the cut-free completeness.
\begin{proposition}[{Okada \cite[Lemma 3.6]{Okada2002}}]\label{OkadaIMALL}
    Let $(M, D_M, V)$ be the syntactic intuitionistic phase model. For any formula $A \in \mathcal{L}$, we have $[A] \in \llbracket A \rrbracket^V \subseteq \mathsf{Pr}_{cf}(A)$, where ``$[A]$'' denotes the singleton multiset of $A$.
\end{proposition}
\begin{proof}
We show (i) $[A] \in \llbracket A \rrbracket^V$ and (ii) $\llbracket A \rrbracket^V \subseteq \mathsf{Pr}_{cf}(A)$ by simultaneous induction on the complexity of $A$.  The crucial case is when $A = A_0 \multimap A_1$. In what follows, we focus exclusively on this case, where the proof of (i) relies on the induction hypothesis for (ii), and conversely, the proof of (ii) utilizes the induction hypothesis for (i).

    (i) We show that $[A_0 \multimap A_1]$ $\in \llbracket A_0 \multimap A_1 \rrbracket^V$ $=$ $\{\Delta \in M \mid \Gamma \cup \Delta \in \llbracket A_1 \rrbracket^V \text{ for all $\Gamma \in \llbracket A_0 \rrbracket^V$}\}$. Fix any $\Gamma \in \llbracket A_0 \rrbracket^V$. We show that $\Gamma  \cup [A_0 \multimap A_1]  \in \llbracket A_1 \rrbracket^V$. By   $\llbracket A_1 \rrbracket^V \in D_M$, we can find some $\mathbb{F} \subseteq \mathcal{L}$ such that $\llbracket A_1 \rrbracket^V = \bigcap\{\mathsf{Pr}_{cf}(C) \mid C \in \mathbb{F}\}$. Fix any $C \in \mathbb{F}$. It suffices to show that $\Gamma \cup [A_0 \multimap A_1]  \in \mathsf{Pr}_{cf}(C)$, i.e., $\Gamma , A_0 \multimap A_1\vdash C$ is provable in $\mathbf{IMALL}^-$. By induction hypothesis for (ii), we have that $\llbracket A_0 \rrbracket^V \subseteq \mathsf{Pr}_{cf}(A_0)$. It follows from $\Gamma \in \llbracket A_0 \rrbracket^V$ that $\Gamma \in \mathsf{Pr}_{cf}(A_0)$, i.e., $\Gamma \vdash A_0$ is provable in $\mathbf{IMALL}^-$. Moreover, by induction hypothesis for (i), we obtain $[A_1] \in \llbracket A_1 \rrbracket^V = \bigcap\{\mathsf{Pr}_{cf}(C) \mid C \in \mathbb{F}\}$. Therefore, $[A_1] \in \mathsf{Pr}_{cf}(C)$, i.e., $A_1 \vdash C$ is provable in $\mathbf{IMALL}^-$. By applying the $(\multimap\vdash)$ rule to $\Gamma \vdash A_0$ and $A_1 \vdash C$, we obtain a proof of $\Gamma, A_0 \multimap A_1, \vdash C$ in $\mathbf{IMALL}^-$.

    (ii) We need to show that $\llbracket A_0 \multimap A_1 \rrbracket^V \subseteq \mathsf{Pr}_{cf}(A_0 \multimap A_1)$. Fix any $\Delta \in M$ and suppose that $\Gamma \cup \Delta \in \llbracket A_1 \rrbracket^V$ for all $\Gamma \in \llbracket A_0 \rrbracket^V$.  We show that $\Delta \in \mathsf{Pr}_{cf}(A_0 \multimap A_1)$, i.e., $\Delta \vdash A_0 \multimap A_1$ is provable in $\mathbf{IMALL}^-$. By induction hypothesis for (i), we have that $[A_0] \in \llbracket A_0 \rrbracket^V$. 
    By the initial supposition, we get $[A_{0}] \cup \Delta 
    \in \llbracket A_1 \rrbracket^V$. 
    Moreover, by induction hypothesis for (ii), we have that $\llbracket A_1 \rrbracket^V \subseteq \mathsf{Pr}_{cf}(A_1)$. 
    Therefore, $[A_0] \cup \Delta \in \mathsf{Pr}_{cf}(A_1)$, i.e., $A_0, \Delta \vdash A_1$ is provable in $\mathbf{IMALL}^-$. By applying the $(\vdash \multimap)$ rule to this sequent, we obtain a proof of $\Delta \vdash A_0 \multimap A_1$ in $\mathbf{IMALL}^-$.
\end{proof}

\if0
\begin{remark}
    In the case of $\multimap$, we used (ii) in induction hypothesis to show (i), and we used (i) in induction hypothesis to show (ii). This case demonstrates the effect of simultaneous induction.
\end{remark}
\fi

We obtain the cut-free completeness immediately.
\begin{proposition}[Cut-free Completeness for $\mathbf{IMALL}$]\label{CutFreeCompletenessForIMALL}
    If a formula $A \in \mathcal{L}$ is true in any intuitionistic phase model, then $\vdash A$ is provable in $\mathbf{IMALL}^-$.
\end{proposition}
\begin{proof}
    Assume that $A$ is true in any intuitionistic phase model. Then, $A$ is true in the syntactic intuitionistic phase model $(M, D_M, V)$, meaning $1 = \emptyset \in \llbracket A \rrbracket^V$. By Lemma \ref{OkadaIMALL}, we get  $\llbracket A \rrbracket^V \subseteq \mathsf{Pr}_{cf}(A)$. Therefore, we conclude that $\emptyset \in \mathsf{Pr}_{cf}(A)$, which implies that $\vdash A$ is provable in $\mathbf{IMALL}^-$.
\end{proof}

By combining the soundness and the cut-free completeness, we can prove the cut-elimination semantically.
\begin{proposition}[Cut-elimination for $\mathbf{IMALL}$]
    If a sequent $\vdash A$ is provable in $\mathbf{IMALL}$, then it is provable in $\mathbf{IMALL}^-$.
\end{proposition}
\begin{proof}
    Assume that $\vdash A$ is provable in $\mathbf{IMALL}$. By Proposition \ref{SoundnessForIMALL}, the formula $A$ is true in any intuitionistic phase model. Then, by Proposition \ref{CutFreeCompletenessForIMALL}, the sequent $\vdash A$ is provable in $\mathbf{IMALL}^-$.
\end{proof}

This theorem can be extended to general sequents by inversion lemmas. 
\begin{corollary}
    If a sequent $\Gamma \vdash C$ is provable in $\mathbf{IMALL}$, then it is provable in $\mathbf{IMALL}^-$.
\end{corollary}
\begin{proof}
    Assume that $\Gamma \vdash C$ is provable in $\mathbf{IMALL}$. By applying the ($\otimes{\vdash}$) rules several times and the ($\vdash\multimap$) rule once, the sequent $\vdash \bigotimes \Gamma \multimap C$ is provable in $\mathbf{IMALL}$. By Proposition \ref{SoundnessForIMALL}, the formula $\bigotimes \Gamma \multimap C$ is true in any intuitionistic phase model. By Proposition \ref{CutFreeCompletenessForIMALL}, the sequent $\vdash \bigotimes \Gamma \multimap C$ is provable in $\mathbf{IMALL}^-$. Then, by Lemma \ref{inversionForIMALL} (Inversion in $\mathbf{IMALL}^-$), $\Gamma \vdash C$ is provable in $\mathbf{IMALL}^-$.
\end{proof}
\section{Syntax and Proof System for $\mu$IMALL}\label{sec:SyntaxAndProofSystemForIMALLWithLeastAndGreatestFixedPoints}
\subsection{Syntax}\label{subsec:Syntax}
Let $\mathcal{V}$ be a countably infinite set of variables and $\mathcal{A}$ be a countably infinite set of atomic formulas such that $\mathcal{V} \cap \mathcal{A} = \emptyset$. Syntax $\mathcal{L}_{\mu}$ of $\mu\mathbf{IMALL}$ is defined as follows:
\[
    A\Coloneqq x \mid a \mid \mathbf{1} \mid \top \mid \mathbf{0} \mid A \otimes A \mid A \with A \mid A \oplus A \mid A \multimap A \mid \mu x.A \mid \nu x.A,
\]
where $x$ is an arbitrary element of $\mathcal{V}$, $a$ is an arbitrary element of $\mathcal{A}$, and the fixpoint operators $\mu$ and $\nu$ bind the occurrences of the variable $x$ in $A$. A variable occurrence that is not bound in $A$ is said to be \emph{free} in $A$. We denote by $\mathsf{FV}(A)$ the set of all the variables that occur freely in $A$. In this paper we always impose the following condition on the syntax:
\begin{quote}
    $\mu x.A$ and $\nu x.A$ are  defined if and only if $x$ is positive in $A$,
\end{quote}
where a free occurrence of $x$ is \emph{positive} (or \emph{negative}) in $A$ if, to reach $x$ in $A$, one traverses the antecedent of $\multimap$ an even (or odd) number of times, and a variable $x$ is \emph{positive} (or \emph{negative}) in $A$ if all the free occurrences of $x$ in $A$ are positive (or negative, respectively). Although the positivity condition is not present in Baelde \cite{Baelde2012}, it is necessary for Lemma \ref{Functoriality} (Functoriality) and Lemma \ref{Monotonicity} (Monotonicity) below. Moreover, in this paper, we identify \emph{$\alpha$-equivalent} formulas; that is, we consistently rename bound variables. Thus, we identify $\mu x.A$ and $\mu y.A(y/x)$, $\nu x.A$ and $\nu y.A(y/x)$. 

There are two remarks on our syntax. Firstly, our definition is the same as the notion of positivity of a variable given in Clairambault \cite[Definition 2.1]{Clairambault2013} in that a bound variable can occur on the left-hand side of the implication in a formula \footnote{In Clairambault \cite[Definition 2.1]{Clairambault2013}, the occurrences of bound variables are said to be "strictly positive" if no occurrence of a bound variable appears on the left-hand side of an implication. We do not impose this strict positivity condition.}. For example, $\mu x. ((x \multimap a) \multimap a)$ is a well-formed formula.
Secondly, De et al.~\cite[Definition 19]{DeJafarrahmaniSaurin2022} refer to the elements of syntax as \emph{pre-formulas}, distinguishing them from \emph{formulas}, which are pre-formulas without free variables.
In contrast, we do not adopt this distinction; all elements of $\mathcal{L}_\mu$ are referred to simply as \emph{formulas}.
Since our formulation enables us to derive all the desired theorems, we have chosen this simpler, more natural definition.

To conclude this section, we explicitly define simultaneous substitution as follows.
\begin{definition}
    Let $A \in \mathcal{L}_\mu$, $\vec{x} = x_1, \ldots, x_n$ be an $n$-tuple of distinct variables, $\vec{F} = F_1, \ldots, F_n$ be an $n$-tuple of formulas. We define a simultaneous substitution $A(\vec{F}/\vec{x})$ by induction on the complexity of $A$ as follows:
    \begin{itemize}
        \item if $A$ is an atomic formula, $\mathbf{1}$, $\top$ or $\mathbf{0}$, then $A(\vec{F}/\vec{x}) = A$,
        \item if $A$ is a variable $y$ and $y \not\in \{x_1, \ldots, x_n\}$, then $A(\vec{F}/\vec{x}) = y(\vec{F}/\vec{x}) = y$,
        \item if $A = x_i$ for some $i$ such that $1 \leq i \leq n$, then $A(\vec{F}/\vec{x}) =x_i(\vec{F}/\vec{x}) = F_i$,
        \item $(A_0 \circ A_1)(\vec{F}/\vec{x}) = A_0(\vec{F}/\vec{x}) \circ A_1(\vec{F}/\vec{x})$, where $\circ \in \{\otimes, \with, \oplus, \multimap\}$,
        \item if $A$ is of the form $\eta y.A'$ $(\eta \in \{\mu, \nu\})$, we can assume without loss of generality that $y \neq x_i$ and $y \not\in \mathsf{FV}(F_i)$ for each $F_i$, and $(\eta y.A')(\vec{F}/\vec{x}) = \eta y.A'(\vec{F}/\vec{x})$.
    \end{itemize}
\end{definition}
An ordinary substitution $A(F/x)$ is a special case of a simultaneous substitution.

\begin{remark} From a broader perspective, our aim when studying $\mu\mathbf{IMALL}$ is to provide a type system for functional programming languages, inherently rooted in a strong logical approach.
Seen as a specification logic, $\mu\mathbf{IMALL}$ enables one to type, via a form of Curry-Howard correspondence whose precise study is left for future work, $\lambda$-terms with recursion normalising to words or trees for instance.
In $\mu\mathbf{IMALL}$ one can for instance specify types for terms normalising to finite or infinite words over an alphabet $\{a,b\}$.
If we choose to model words as tensors of atomic formulas corresponding to letters, then a type for terms normalising to $a^n b^\omega$ would be $\phi_1\,=\,(\mu x.((a \otimes x) \oplus \mathbf{1})) \otimes (\nu y.(b \otimes y))$.
Infinite words over that alphabet that do not contain an infinite sequence of $a$s would be modeled by $\phi_2\,=\,\nu x.((\mu y.((a \otimes y) \oplus \mathbf{1})) \otimes b \otimes x)$.
\end{remark}

\subsection{Sequent Calculus}\label{subsec:SequentCalculus}

A \emph{sequent} $\Gamma \vdash C$ is defined as in the same way as $\mathbf{IMALL}$. Sequent calculus $\mu\mathbf{IMALL}$ is obtained by adding the following rules to $\mathbf{IMALL}$:
\[
    \begin{bprooftree}
        \AxiomC{$\Gamma \vdash A(\mu x.A/x)$}
        \RightLabel{$({\vdash}\mu)$}
        \UnaryInfC{$\Gamma \vdash \mu x.A$}
    \end{bprooftree},
    \begin{bprooftree}
        \AxiomC{$A(S/x)\vdash S$}
        \AxiomC{$S, \Gamma \vdash C$}
        \RightLabel{$(\mu{\vdash})$}
        \BinaryInfC{$\mu x.A, \Gamma \vdash C$}
    \end{bprooftree}
\]
\[
    \begin{bprooftree}
        \AxiomC{$\Gamma \vdash S$}
        \AxiomC{$S\vdash A(S/x)$}
        \RightLabel{$({\vdash}\nu)$}
        \BinaryInfC{$\Gamma \vdash \nu x.A$}
    \end{bprooftree},
    \begin{bprooftree}
        \AxiomC{$A(\nu x.A/x), \Gamma \vdash C$}
        \RightLabel{$(\nu{\vdash})$}
        \UnaryInfC{$\nu x.A, \Gamma \vdash C.$}
    \end{bprooftree}
\]
We denote $\mu\mathbf{IMALL}$ without $(Cut)$ by $\mu\mathbf{IMALL}^-$. 
Lemma \ref{inversionForIMALL} (Inversion) extends to the system with fixpoint rule as follows. 
\begin{lemma}[Inversion in $\mu\mathbf{IMALL}^-$]\label{inversion}
    \begin{enumerate}
        \item if $\Gamma \vdash A \multimap B$ is provable in $\mu\mathbf{IMALL}^-$, then $A, \Gamma \vdash B$ is also provable in $\mu\mathbf{IMALL}^-$.
        \item if $A \otimes B, \Gamma \vdash C$ is provable in $\mu\mathbf{IMALL}^-$, then $A, B, \Gamma \vdash C$ is also provable in $\mu\mathbf{IMALL}^-$.
    \end{enumerate}
\end{lemma}
\subsection{Functoriality}\label{subsec:Functoriality}
This section establishes that 
the following \emph{functoriality} rule is admissible in $\mu\mathbf{IMALL}$:
\[
    \begin{bprooftree}
        \AxiomC{$F \vdash G$}
        \RightLabel{$(func)$}
        \UnaryInfC{$A(F/x)\vdash A(G/x)$}
    \end{bprooftree}
\]
where $x$ is positive in $A$, i.e., all the occurrences of $x$ in $A$ are positive. This rule is introduced to first-order classical linear logic with the fixpoints with no restrictions on variable occurrence in Baelde \cite{Baelde2012}.

% To show this, we define an extended version of substitution.
% \begin{definition}
%     Let $A \in \mathcal{L}_\mu$, $\vec{x} = x_1, \ldots, x_n$ be an $n$-tuple of distinct variables in $A$, $\vec{F} = F_1, \ldots, F_n$ and $\vec{G} = G_1, \ldots, G_n$ be $n$-tuples of formulas. We define an extended simultaneous substitution $\sigma (A, \vec{x}, \vec{F}, \vec{G})$ on the complexity of $A$ as follows:
%     \begin{itemize}
%         \item if $A$ is an atom, $\mathbf{1}$, $\top$ or $\mathbf{0}$, then $\sigma (A, \vec{x}, \vec{F}, \vec{G}) = A$,
%         \item if $A$ is a variable $y$ and $y$'s occurrence is not in $\{x_1, \ldots, x_n\}$, then $\sigma (y, \vec{x}, \vec{F}, \vec{G}) = y$,
%         \item if $A = x_i$ for some $i$ such that $1 \leq i \leq n$, then $\sigma (x_i, \vec{x}, \vec{F}, \vec{G}) = F_i$,
%         \item $\sigma (A_0 \circ A_1, \vec{x}, \vec{F}, \vec{G}) = \sigma (A_0, \vec{x}, \vec{F}, \vec{G}) \circ \sigma (A_1, \vec{x}, \vec{F}, \vec{G})$, where $\circ \in \{\otimes, \with, \oplus\}$,
%         \item $\sigma (A_0 \multimap A_1, \vec{x}, \vec{F}, \vec{G}) = \sigma (A_0, \vec{x}, \vec{G}, \vec{F}) \multimap \sigma (A_1, \vec{x}, \vec{F}, \vec{G})$,
%         \item for $\eta y.A'$, we can assume that $y \neq x_i$ for any $i$ such that $1 \leq i \leq n$, and then $\sigma (\eta y.A', \vec{x}, \vec{F}, \vec{G}) = \eta y.\sigma (A', \vec{x}, \vec{F}, \vec{G})$.
%     \end{itemize}
% \end{definition}
% Through this substitution, if $x_i$ is a positive occurrence in $A$, then it is replaced by $F_i$, otherwise, by $G_i$.
We impose the above positivity restriction in order to handle linear implication $\multimap$. Under this restriction, we cannot take $x \multimap x$ as $A$ in the functoriality rule since the variable $x$ occurs both positively and negatively. In contrast, a formula $y \multimap x$ with $x \neq y$ is permitted, as every occurrence of $x$ in $A$ is positive. Let us now establish the admissibility of the ($func$) rule.
\begin{lemma}[Functoriality]\label{Functoriality}
    Let $A, F, G \in \mathcal{L}_\mu$ and $x \in \mathcal{V}$. If $F \vdash G$ is provable in $\mu\mathbf{IMALL}$ and $x$ is positive in $A$, then $A(F/x)\vdash A(G/x)$ is also provable.
\end{lemma}
\begin{proof}
    In what follows, we use ``provable'' to mean  ``provable in $\mu\mathbf{IMALL}$.'' By induction on the complexity of $A$, we show the following more general statement: 
    \begin{quotation}
    \noindent for any $n \in \mathbb{N}$, $n$-tuples of distinct variables $\vec{x} = x_1, \ldots, x_n$ and $\vec{y} = y_1, \ldots, y_n$, $n$-tuples of formulas $\vec{F} = F_1, \ldots, F_n$ and $\vec{G} = G_1, \ldots, G_n$, if $F_1 \vdash G_1, \ldots, F_n \vdash G_n$ are provable, $x_1, \ldots, x_n$ are all positive in $A$, and $y_1, \ldots, y_n$ are all negative in $A$, then $A(\vec{F}, \vec{G}/\vec{x}, \vec{y}) \vdash A(\vec{F}, \vec{G}/\vec{y}, \vec{x})$ is provable. 
    \end{quotation}
    In particular, if $x$ is positive and $y$ does not occur in $A$, which implies that $y$ is trivially negative in $A$, we obtain that $A(F, G/x, y)\vdash A(F, G/y, x) \equiv A(F/x)\vdash A(G/x)$ is provable. 
        When we take $y_2 \multimap x_1$ as an example of $A$ and assume the provability of $F_{1} \vdash G_{1}$ and $F_{2} \vdash G_{2}$, then we obtain the provability of $G_{2} \multimap F_{1} \vdash F_{2} \multimap G_{1}$ by the general statement above (where it is noted that $y_2$ is negative in $y_2 \multimap x_1$).    
    For the base step, we proceed as follows. 
    \begin{itemize}
       \item Let $A$ be an atomic formula, $\mathbf{1}$, $\top$ or $\mathbf{0}$. 
       Let $\vec{F}$,  $\vec{G}$, $\vec{x}$ and $\vec{y}$ satisfy the required conditions.     
       %Fix any $n \in \mathbb{N}$, $n$-tuples of distinct variables $\vec{x} = x_1, \ldots, x_n$ and $\vec{y} = y_1, \ldots, y_n$, $n$-tuples of formulas $\vec{F} = F_1, \ldots, F_n$ and $\vec{G} = G_1, \ldots, G_n$. Assume that $F_1 \vdash G_1, \ldots, F_n \vdash G_n$ are provable, and $x_1, \ldots, x_n$ are positive, $y_1, \ldots, y_n$ are negative in $A$. 
       Since $A(\vec{F}, \vec{G}/\vec{x}, \vec{y}) = A(\vec{F}, \vec{G}/\vec{y}, \vec{x}) = A$, we need to show that $A \vdash A$ is provable, but this is clear.
        \item Let $A = z \in \mathcal{V}$ and let $\vec{F}$,  $\vec{G}$, $\vec{x}$ and $\vec{y}$ satisfy the required conditions.
        %Fix any $n \in \mathbb{N}$, $n$-tuples of distinct variables $\vec{x} = x_1, \ldots, x_n$ and $\vec{y} = y_1, \ldots, y_n$, $n$-tuples of formulas $\vec{F} = F_1, \ldots, F_n$ and $\vec{G} = G_1, \ldots, G_n$. Assume that $F_1 \vdash G_1, \ldots, F_n \vdash G_n$ are provable, and $x_1, \ldots, x_n$ are positive, $y_1, \ldots, y_n$ are negative in $z$. 
        Since the case in which $z \not\in \{x_1, \ldots, x_n, y_1, \ldots, y_n\}$ can be shown similarly to the first case, we assume $z \in \{x_1, \ldots, x_n, y_1, \ldots, y_n\}$. Then, since $z$ is positive in $z$, we have $z = x_i$ for some $i$ such that $1 \leq i \leq n$. Thus, we obtain $z(\vec{F}, \vec{G}/\vec{x}, \vec{y}) = F_i$ and $z(\vec{F}, \vec{G}/\vec{y}, \vec{x}) = G_i$. We need to show that $F_i \vdash G_i$ is provable, but this is one of the assumptions.
    \end{itemize}
    For the inductive step, our argument proceeds as follows.
    \begin{itemize}
        \item Let $A = A_0 \otimes A_1$ and  
        let $\vec{F}$,  $\vec{G}$, $\vec{x}$ and $\vec{y}$ satisfy the required conditions.
        %Fix any $n \in \mathbb{N}$, $n$-tuples of distinct variables $\vec{x} = x_1, \ldots, x_n$ and $\vec{y} = y_1, \ldots, y_n$, $n$-tuples of formulas $\vec{F} = F_1, \ldots, F_n$ and $\vec{G} = G_1, \ldots, G_n$. Assume that $F_1 \vdash G_1, \ldots, F_n \vdash G_n$ are provable, and $x_1, \ldots, x_n$ are positive, $y_1, \ldots, y_n$ are negative in $A$. 
        Since $(A_0 \otimes A_1)(\vec{F}, \vec{G}/\vec{x}, \vec{y}) = A_0(\vec{F}, \vec{G}/\vec{x}, \vec{y}) \otimes A_1(\vec{F}, \vec{G}/\vec{x}, \vec{y})$ and $(A_0 \otimes A_1)(\vec{F}, \vec{G}/\vec{y}, \vec{x}) = A_0(\vec{F}, \vec{G}/\vec{y}, \vec{x}) \otimes A_1(\vec{F}, \vec{G}/\vec{y}, \vec{x})$, we need to show that $A_0(\vec{F}, \vec{G}/\vec{x}, \vec{y}) \otimes A_1(\vec{F}, \vec{G}/\vec{x}, \vec{y}) \vdash A_0(\vec{F}, \vec{G}/\vec{y}, \vec{x}) \otimes A_1(\vec{F}, \vec{G}/\vec{y}, \vec{x})$ is provable. By induction hypothesis, $A_0(\vec{F}, \vec{G}/\vec{x}, \vec{y}) \vdash A_0(\vec{F}, \vec{G}/\vec{y}, \vec{x})$ and $A_1(\vec{F}, \vec{G}/\vec{x}, \vec{y}) \vdash A_1(\vec{F}, \vec{G}/\vec{y}, \vec{x})$ are provable. From these we obtain the following proof:
        \[
            \begin{bprooftree}
                \AxiomC{$A_0(\vec{F}, \vec{G}/\vec{x}, \vec{y}) \vdash A_0(\vec{F}, \vec{G}/\vec{y}, \vec{x})$}
                \AxiomC{$A_1(\vec{F}, \vec{G}/\vec{x}, \vec{y}) \vdash A_1(\vec{F}, \vec{G}/\vec{y}, \vec{x})$}
                \RightLabel{(${\vdash}\otimes$)}
                \BinaryInfC{$A_0(\vec{F}, \vec{G}/\vec{x}, \vec{y}), A_1(\vec{F}, \vec{G}/\vec{x}, \vec{y}) \vdash A_0(\vec{F}, \vec{G}/\vec{y}, \vec{x}) \otimes A_1(\vec{F}, \vec{G}/\vec{y}, \vec{x})$}
                \RightLabel{($\otimes{\vdash}$)}
                \UnaryInfC{$A_0(\vec{F}, \vec{G}/\vec{x}, \vec{y}) \otimes A_1(\vec{F}, \vec{G}/\vec{x}, \vec{y}) \vdash A_0(\vec{F}, \vec{G}/\vec{y}, \vec{x}) \otimes A_1(\vec{F}, \vec{G}/\vec{y}, \vec{x})$.}
            \end{bprooftree}
        \]
        The cases where $A = A_0 \with A_1$ or $A = A_0 \oplus A_1$ can be shown similarly.
        \item Let $A = A_0 \multimap A_1$ and let $\vec{F}$,  $\vec{G}$, $\vec{x}$ and $\vec{y}$ satisfy the required conditions.
        %Fix any $n \in \mathbb{N}$, $n$-tuples of distinct variables $\vec{x} = x_1, \ldots, x_n$ and $\vec{y} = y_1, \ldots, y_n$, $n$-tuples of formulas $\vec{F} = F_1, \ldots, F_n$ and $\vec{G} = G_1, \ldots, G_n$. Assume that $F_1 \vdash G_1, \ldots, F_n \vdash G_n$ are provable, and $x_1, \ldots, x_n$ are positive, $y_1, \ldots, y_n$ are negative in $A$. 
        As $(A_0 \multimap A_1)(\vec{F}, \vec{G}/\vec{x}, \vec{y}) = A_0(\vec{F}, \vec{G}/\vec{x}, \vec{y}) \multimap A_1(\vec{F}, \vec{G}/\vec{x}, \vec{y})$ and $(A_0 \multimap A_1)(\vec{F}, \vec{G}/\vec{y}, \vec{x}) = A_0(\vec{F}, \vec{G}/\vec{y}, \vec{x}) \multimap A_1(\vec{F}, \vec{G}/\vec{y}, \vec{x})$, we need to show that $A_0(\vec{F}, \vec{G}/\vec{x}, \vec{y}) \multimap A_1(\vec{F}, \vec{G}/\vec{x}, \vec{y}) \vdash A_0(\vec{F}, \vec{G}/\vec{y}, \vec{x}) \multimap A_1(\vec{F}, \vec{G}/\vec{y}, \vec{x})$ is provable. 
        Recall that 
        $x_1, \ldots, x_n$ are positive in $A_{0} \multimap A_{1}$ and 
        $y_1, \ldots, y_n$ are negative in $A_{0} \multimap A_{1}$.
        Observe that $x_1, \ldots, x_n$ are positive in $A_1$, $y_1, \ldots, y_n$ are negative in $A_1$.        
        By induction hypothesis, $A_1(\vec{F}, \vec{G}/\vec{x}, \vec{y}) \vdash A_1(\vec{F}, \vec{G}/\vec{y}, \vec{x})$ is provable. 
        Observe also that $y_1, \ldots, y_n$ are positive in $A_0$ and $x_1, \ldots, x_n$ are negative in $A_0$. 
        Again by induction hypothesis, $A_0(\vec{F}, \vec{G}/\vec{y}, \vec{x}) \vdash A_0(\vec{F}, \vec{G}/\vec{x}, \vec{y})$ is provable. 
        From these we obtain the following proof:
        \[
            \begin{bprooftree}
                \AxiomC{$A_0(\vec{F}, \vec{G}/\vec{y}, \vec{x}) \vdash A_0(\vec{F}, \vec{G}/\vec{x}, \vec{y})$}
                \AxiomC{$A_1(\vec{F}, \vec{G}/\vec{x}, \vec{y}) \vdash A_1(\vec{F}, \vec{G}/\vec{y}, \vec{x})$}
                \RightLabel{($\multimap\vdash$)}
                \BinaryInfC{$A_0(\vec{F}, \vec{G}/\vec{x}, \vec{y}) \multimap A_1(\vec{F}, \vec{G}/\vec{x}, \vec{y}), A_0(\vec{F}, \vec{G}/\vec{y}, \vec{x}) \vdash A_1(\vec{F}, \vec{G}/\vec{y}, \vec{x})$}
                \RightLabel{($\vdash\multimap$)}
                \UnaryInfC{$A_0(\vec{F}, \vec{G}/\vec{x}, \vec{y}) \multimap A_1(\vec{F}, \vec{G}/\vec{x}, \vec{y}) \vdash A_0(\vec{F}, \vec{G}/\vec{y}, \vec{x}) \multimap A_1(\vec{F}, \vec{G}/\vec{y}, \vec{x})$.}
            \end{bprooftree}
        \]
        \item Let $A = \mu z.B$ and 
        let $\vec{F}$,  $\vec{G}$, $\vec{x}$ and $\vec{y}$ satisfy the required conditions. We can assume without loss of generality that $z \not\in \{x_1, \ldots, x_n, y_1, \ldots, y_n\}$, $z \not\in \mathsf{FV}(F_i)$ for each $F_i$, and $z \not\in \mathsf{FV}(G_i)$ for each $G_i$. Since $(\mu z.B)(\vec{F}, \vec{G}/\vec{x}, \vec{y}) = \mu z.B(\vec{F}, \vec{G}/\vec{x}, \vec{y})$ and $(\mu z.B)(\vec{F}, \vec{G}/\vec{y}, \vec{x}) = \mu z.B(\vec{F}, \vec{G}/\vec{y}, \vec{x})$, we need to show that $\mu z.B(\vec{F}, \vec{G}/\vec{x}, \vec{y}) \vdash \mu z.B(\vec{F}, \vec{G}/\vec{y}, \vec{x})$ is provable. Take $w \in \mathcal{V}$ such that $w \not\in \mathsf{FV}(B)$. Observe that $w$ is trivially negative in $B$. Write $E := \mu z.B(\vec{F}, \vec{G}/\vec{y}, \vec{x})$.  Since $\mu z.B$ is defined, $z$ is positive in $B$. Moreover, it is clear that $E \vdash E$ is provable. Therefore, by 
        noting that both $B(\vec{F}, E, \vec{G}, E/\vec{x}, z, \vec{y}, w) = B(\vec{F}, \vec{G}/\vec{x}, \vec{y})(E/z)$ and $B(\vec{F}, E, \vec{G}, E/\vec{y}, w, \vec{x}, z) = B(\vec{F}, \vec{G}/\vec{y}, \vec{x})(E/z)$, 
        induction hypothesis tells us that the following sequent   $B(\vec{F}, \vec{G}/\vec{x}, \vec{y})(E/z) \vdash B(\vec{F}, \vec{G}/\vec{y}, \vec{z})(E/z)$ is provable. From this we obtain the following proof:
        \[
            \begin{bprooftree}
                \AxiomC{$B(\vec{F}, \vec{G}/\vec{x}, \vec{y})(E/z) \vdash B(\vec{F}, \vec{G}/\vec{y}, \vec{x})(E/z)$}
                \RightLabel{(${\vdash}\mu$)}
                \UnaryInfC{$B(\vec{F}, \vec{G}/\vec{x}, \vec{y})(E/z) \vdash \mu z.B(\vec{F}, \vec{G}/\vec{y}, \vec{x}) (= E)$}
                \AxiomC{$E \vdash E$}
                \RightLabel{($\mu{\vdash}$)}
                \BinaryInfC{$\mu z.B(\vec{F}, \vec{G}/\vec{x}, \vec{y}) \vdash \mu z.B(\vec{F}, \vec{G}/\vec{y}, \vec{x})$,}
            \end{bprooftree}
        \]
        where we note that $\Gamma$ is an empty multiset $\varnothing$, $S$ is $E$ and $C$ is also $E$ in the the rule $(\mu{\vdash})$. 
        The case where $A = \nu z.B$ can be shown similarly. \qedhere
    \end{itemize}
\end{proof}
\begin{remark}
    In Clairambault \cite[Definitions 2.2 and 2.3]{Clairambault2013}, two types of functors and two types of functoriality rules are defined: positive and negative. 
    Similarly, 
    our proof of the Functoriality Lemma requires handling both positive and negative variables simultaneously. 
    %This necessitates dealing with two types of functoriality within a single setting.
\end{remark}
\section{Phase Semantics for $\mu$IMALL and Soundness}\label{sec:PhaseSemanticsForMuIMALL}
\subsection{$\mu$-Phase Model}

Let $(M, D_M)$ be an intuitionistic phase space. To define a phase model for $\mu\mathbf{IMALL}$, we restrict the codomain of valuations to a specific collection of closed sets, following the approach of De et al.~\cite{DeJafarrahmaniSaurin2022}. This restriction is crucial for establishing cut-free completeness. While De et al. employed double negation as their closure operator, we utilize a closure function $cl$ induced by $D_M$.

\begin{definition}
An intuitionistic $\mu$-phase model is a quadruple $(M, D_M, \mathcal{D}, V)$ where $(M, D_M)$ is an intuitionistic phase space, $\mathcal{D} \subseteq D_M$, and a function $V$ is a \emph{$\mathcal{D}$-valuation} whose domain is $\mathcal{V} \cup \mathcal{A}$ and whose codomain is $\mathcal{D}$. 
\end{definition}

\if0
Let $(M, D_M, V)$ be an intuitionistic phase model. In order to define a phase model for $\mu\mathbf{IMALL}$, we need to make a restriction of the range of a set of closed sets as De et al. \cite{DeJafarrahmaniSaurin2022} did. Therefore, let $\mathcal{D} \subseteq D_M$. A \emph{$\mathcal{D}$-valuation} is a function $V$ whose domain is $\mathcal{V} \cup \mathcal{A}$ and codomain is $\mathcal{D}$.
\fi

\noindent Then we define the \emph{interpretation} 
$\llbracket A \rrbracket^V$ of a formula $A \in \mathcal{L}_\mu$ as in  $\mathbf{IMALL}$ except the following clauses for variables, atoms, and formulas with the fixpoint operators: for $p \in \mathcal{V} \cup \mathcal{A}$ 
we define $\llbracket p\rrbracket^V = V(p)$,
and for $\mu x.A'$ and $\nu x.A'$, we define
\begin{align*}
    \llbracket \mu x.A'\rrbracket^V &= \bigcap\left\{X \in \mathcal{D} \mid \llbracket A'\rrbracket^{V[x \mapsto X]} \subseteq X \right\},& \llbracket \nu x.A'\rrbracket^V &= cl\left(\bigcup\left\{X \in \mathcal{D} \mid X \subseteq \llbracket A'\rrbracket^{V[x \mapsto X]}\right\}\right),
\end{align*}
where $V[x \mapsto X]$ is defined by
\[
    V[x \mapsto X](p) = \begin{cases}
        X & \text{if $p = x$},\\
        V(p) & \text{{otherwise}}.
    \end{cases}
\]
For any formula $A \in \mathcal{L}_\mu$, $\llbracket A\rrbracket^{V} \in D_{M}$ holds. 
This suggests the following definition. 

\begin{definition}
Let $(M, D_M, \mathcal{D}, V)$ be 
an intuitionistic $\mu$-phase model.  
We say that 
a $\mathcal{D}$-valuation $V$ is \emph{admissible} if $\llbracket A \rrbracket^V \in \mathcal{D}$ for all formulas $A \in \mathcal{L}_\mu$. 
An intuitionistic $\mu$-phase model $(M, D_M, \mathcal{D}, V)$ is \emph{admissible} if $V$ is {admissible}.
\end{definition}

Note that when $\mathcal{D} = D_M$, an intuitionistic $\mu$-phase model $(M, D_M, \mathcal{D}, V)$ is always admissible. Similarly to De et al. \cite[Lemma 27]{DeJafarrahmaniSaurin2022}, the following monotonicity lemma can be shown, with a restriction to positive occurrences of variables.
\begin{lemma}[Monotonicity]\label{Monotonicity}
    Let $(M, D_M, \mathcal{D}, V)$ be an admissible intuitionistic $\mu$-phase model, and $X, Y \in \mathcal{D}$. If $x \in \mathcal{V}$ is positive in a formula $A \in \mathcal{L}_\mu$ and  $X \subseteq Y$, then $\llbracket A \rrbracket^{V[x \mapsto X]} \subseteq \llbracket A \rrbracket^{V[x \mapsto Y]}$.
\end{lemma}
\begin{proof}
    Fix any intuitionistic $\mu$-phase space 
    $(M, D_M, \mathcal{D})$. 
    We show by induction on the complexity of $A$ the following: for any admissible $\mathcal{D}$-valuation $V$, $x \in \mathcal{V}$, and $X, Y \in \mathcal{D}$ with $X \subseteq Y$, the following hold: (i) if $x$ is positive in $A$, then $\llbracket A \rrbracket^{V[x \mapsto X]} \subseteq \llbracket A \rrbracket^{V[x \mapsto Y]}$ and (ii) if $x$ is negative in $A$, then $\llbracket A \rrbracket^{V[x \mapsto Y]} \subseteq \llbracket A \rrbracket^{V[x \mapsto X]}$. 
    We show (i) only for the cases where $A = A_0 \multimap A_1$. The remaining cases for statement (i) can be proved similarly to De et al.~\cite[Lemma 27]{DeJafarrahmaniSaurin2022}.
    Furthermore, statement (ii) for each case can be shown by replacing ``positive'' with ``negative'' and reversing the inclusion ``$\subseteq$'' in the proof of (i). Let $A = A_0 \multimap A_1$. Fix any admissible $\mathcal{D}$-valuation $V$, $x \in \mathcal{V}$, and $X, Y \in \mathcal{D}$ such that $X \subseteq Y$. Suppose that $x$ is positive in $A_0 \multimap A_1$. We show the following:
    \begin{align*}
        &\{b \in M \mid \text{for all $a \in \llbracket A_0 \rrbracket^{V[x \mapsto X]}$, $ab \in \llbracket A_1 \rrbracket^{V[x \mapsto X]}$}\}\\
        \subseteq ~&\{b \in M \mid \text{for all $a \in \llbracket A_0 \rrbracket^{V[x \mapsto Y]}$, $ab \in \llbracket A_1 \rrbracket^{V[x \mapsto Y]}$}\}.
    \end{align*}
    Fix any $b \in M$ such that $ab \in \llbracket A_1 \rrbracket^{V[x \mapsto X]}$ for all $a \in \llbracket A_0 \rrbracket^{V[x \mapsto X]}$. 
    Furthermore, fix any $a' \in \llbracket A_0 \rrbracket^{V[x \mapsto Y]}$. We show that $a'b \in \llbracket A_1 \rrbracket^{V[x \mapsto Y]}$. Observe that $x$ is negative in $A_0$ and positive in $A_1$. Therefore, we obtain $a' \in \llbracket A_0 \rrbracket^{V[x \mapsto Y]} \subseteq \llbracket A_0 \rrbracket^{V[x \mapsto X]}$ by induction hypothesis for (ii), and so $a'b \in \llbracket A_1 \rrbracket^{V[x \mapsto X]} \subseteq \llbracket A_1 \rrbracket^{V[x \mapsto Y]}$ by induction hypothesis for (i).
\end{proof}

\subsection{Soundness}
\begin{lemma}[Soundness for $\mu\mathbf{IMALL}$]\label{SoundnessАorMuIMALL}
    If a sequent $A_1, \ldots, A_n \vdash C$ is provable in $\mu\mathbf{IMALL}$, then $\llbracket A_1 \rrbracket^V \cdots \llbracket A_n \rrbracket^V \subseteq \llbracket C \rrbracket^V$ holds for any admissible intuitionistic $\mu$-phase model $(M, D_M, \mathcal{D}, V)$, where the empty product of the monoid operation is defined as $\{1\}$. In particular, if $\vdash A$ is provable in $\mu\mathbf{IMALL}$, then $A$ is true in any admissible intuitionistic $\mu$-phase model.
\end{lemma}
\begin{proof}
Fix an arbitrary admissible intuitionistic $\mu$-phase model $(M, D_M, \mathcal{D}, V)$. We shall prove that $\llbracket A_1 \rrbracket^V \cdots \llbracket A_n \rrbracket^V \subseteq \llbracket C \rrbracket^V$ by induction on the proof of $A_1, \ldots, A_n \vdash C$. The proof proceeds by case analysis on the last rule applied; here, we only present the cases for the fixpoint operators.
    \begin{itemize}
        \item Let the last applied rule be
        \[
            \begin{bprooftree}
                \AxiomC{$A_1, \ldots, A_n \vdash B(\mu x.B/x)$}
                \RightLabel{(${\vdash}\mu$)}
                \UnaryInfC{$A_1, \ldots, A_n \vdash \mu x.B.$}
            \end{bprooftree}
        \]
        We show that $\llbracket A_1 \rrbracket^V \cdots \llbracket A_n \rrbracket^V \subseteq \llbracket \mu x.B \rrbracket^V$. Since $\llbracket A_1 \rrbracket^V \cdots \llbracket A_n \rrbracket^V \subseteq \llbracket B(\mu x.B/x) \rrbracket^V$ by induction hypothesis, it suffices to show $\llbracket B(\mu x.B/x) \rrbracket^V \subseteq \llbracket \mu x.B \rrbracket^V$. Note that $\llbracket B(\mu x.B/x) \rrbracket^V = \llbracket B \rrbracket^{V[x \mapsto \llbracket \mu x.B \rrbracket^V]}$ since $\llbracket \mu x.B \rrbracket^V \in \mathcal{D}$ and so $V[x \mapsto \llbracket \mu x.B \rrbracket^V]$ is a $\mathcal{D}$-valuation. Recall that $\llbracket \mu x.B \rrbracket^V = \bigcap \{X \in \mathcal{D} \mid \llbracket B \rrbracket^{V[x \mapsto X]} \subseteq X\}$. Hence, for our goal, let us fix any $X \in \mathcal{D}$ such that $\llbracket B \rrbracket^{V[x \mapsto X]} \subseteq X$. We show that $\llbracket B \rrbracket^{V[x \mapsto \llbracket \mu x.B \rrbracket^V]} \subseteq X$. By $\llbracket B \rrbracket^{V[x \mapsto X]} \subseteq X$, it suffices to show $\llbracket B \rrbracket^{V[x \mapsto \llbracket \mu x.B \rrbracket^V]} \subseteq \llbracket B \rrbracket^{V[x \mapsto X]}$. But this holds by Lemma \ref{Monotonicity} (Monotonicity) because $x$ is positive in $B$ and $\llbracket \mu x.B \rrbracket^V \subseteq X$.
        \item Let the last applied rule be
        \[
            \begin{bprooftree}
                \AxiomC{$B(S/x) \vdash S$}
                \AxiomC{$S, A_1, \ldots, A_n \vdash C$}
                \RightLabel{($\mu{\vdash}$)}
                \BinaryInfC{$\mu x.B, A_1, \ldots, A_n \vdash C.$}
            \end{bprooftree}
        \]
        We show that $\llbracket \mu x.B \rrbracket^V \llbracket A_1 \rrbracket^V \cdots \llbracket A_n \rrbracket^V \subseteq \llbracket C \rrbracket^V$. Since $\llbracket S \rrbracket^V \llbracket A_1 \rrbracket^V \cdots \llbracket A_n \rrbracket^V \subseteq \llbracket C \rrbracket^V$ by induction hypothesis, it suffices to show that $\llbracket \mu x.B \rrbracket^V \subseteq \llbracket S \rrbracket^V$. To show this, it is enough to show that $\llbracket B \rrbracket^{V[x \mapsto \llbracket S \rrbracket^V]} \subseteq \llbracket S \rrbracket^V$ since $\llbracket S \rrbracket^V \in \mathcal{D}$. This holds as induction hypothesis we obtain $\llbracket B \rrbracket^{V[x \mapsto \llbracket S \rrbracket^V]} = \llbracket B(S/x) \rrbracket^V \subseteq \llbracket S \rrbracket^V$.
        \item Let the last applied rule be
        \[
            \begin{bprooftree}
                \AxiomC{$A_1, \ldots, A_n \vdash S$}
                \AxiomC{$S \vdash B(S/x)$}
                \RightLabel{(${\vdash}\nu$)}
                \BinaryInfC{$A_1, \ldots, A_n \vdash \nu x.B.$}
            \end{bprooftree}
        \]
        We show $\llbracket A_1 \rrbracket^V \cdots \llbracket A_n \rrbracket^V \subseteq \llbracket \nu x.B \rrbracket^V$. Since $\llbracket A_1 \rrbracket^V \cdots \llbracket A_n \rrbracket^V \subseteq \llbracket S \rrbracket^V$ by induction hypothesis, it suffices to show that $\llbracket S \rrbracket^V \subseteq \llbracket \nu x.B \rrbracket^V$. Then, recall that $\bigcup \{X \in \mathcal{D} \mid X \subseteq \llbracket B \rrbracket^{V[x \mapsto X]}\} \subseteq cl(\bigcup \{X \in \mathcal{D} \mid X \subseteq \llbracket B \rrbracket^{V[x \mapsto X]}\}) = \llbracket \nu x.B \rrbracket^V$ by Proposition \ref{closureProposition}. 
        It is noted that $\llbracket S \rrbracket^V \in\mathcal{D}$. 
        Hence, in order to show that $\llbracket S \rrbracket^V \subseteq \llbracket \nu x.B \rrbracket^V$, we prove that $\llbracket S \rrbracket^V \subseteq \llbracket B \rrbracket^{V[x \mapsto \llbracket S \rrbracket^V]}$, which holds by our induction hypothesis $\llbracket S \rrbracket^V \subseteq \llbracket B(S/x) \rrbracket^V = \llbracket B \rrbracket^{V[x \mapsto \llbracket S \rrbracket^V]}$.
        \item Let the last applied rule be
        \[
            \begin{bprooftree}
                \AxiomC{$B(\nu x.B/x), A_1, \ldots, A_n \vdash C$}
                \RightLabel{($\nu{\vdash}$)}
                \UnaryInfC{$\nu x.B, A_1, \ldots, A_n \vdash C.$}
            \end{bprooftree}
        \]
        We show that $\llbracket \nu x.B \rrbracket^V \llbracket A_1 \rrbracket^V \cdots \llbracket A_n \rrbracket^V \subseteq \llbracket C \rrbracket^V$. It suffices to show  $\llbracket \nu x.B \rrbracket^V \subseteq \llbracket B(\nu x.B/x) \rrbracket^V$ since $\llbracket B(\nu x.B/x) \rrbracket^V \llbracket A_1 \rrbracket^V \cdots \llbracket A_n \rrbracket^V \subseteq \llbracket C \rrbracket^V$ by induction hypothesis. To show it, it is enough to show  $\llbracket \nu x.B \rrbracket^V$ = $\bigcup \{X \in \mathcal{D} \mid X \subseteq \llbracket B \rrbracket^{V[x \mapsto X]}\} \subseteq \llbracket B \rrbracket^{V[x \mapsto \llbracket \nu x.B \rrbracket^V]}$. Fix any $X \in \mathcal{D}$ such that $X \subseteq \llbracket B \rrbracket^{V[x \mapsto X]}$. It is noted that $X \subseteq \llbracket \nu x.B \rrbracket^V$. 
        We show $X \subseteq \llbracket B \rrbracket^{V[x \mapsto \llbracket \nu x.B \rrbracket^V]}$. By $X \subseteq \llbracket B \rrbracket^{V[x \mapsto X]}$, it suffices to show  $\llbracket B \rrbracket^{V[x \mapsto X]} \subseteq \llbracket B \rrbracket^{V[x \mapsto \llbracket \nu x.B \rrbracket^V]}$. But this holds by Lemma \ref{Monotonicity} (Monotonicity) and $X \subseteq \llbracket \nu x.B \rrbracket^V$. \qedhere
    \end{itemize}
\end{proof}
\section{Completeness and Semantic Cut-elimination}\label{sec:CompletenessAndSemanticCutElimination}
\subsection{Syntactic Phase Model}\label{subsec:SyntacticModel}

We define the \emph{syntactic intuitionistic $\mu$-phase model} $(M, D_M, \mathcal{D}, V)$ as the syntactic intuitionistic phase model for $\mathbf{IMALL}$ extended with a collection $\mathcal{D}$ of sets as follows.  

\begin{definition}
\label{dfn:syntacticMuPhaseModel}
    For a formula $C$, we define 
    \[
    \mathsf{Pr}_{cf}(C) = \{\Gamma \mid \text{$\Gamma \vdash C$ is provable without $(Cut)$ in $\mu\mathbf{IMALL}$}\}.
    \]
    Define the syntactic intuitionistic $\mu$-phase model $(M, D_M, \mathcal{D}, V)$ as follows:
    \begin{itemize}
        \item The base set $M$ of the monoid is the set of all finite multisets of formulas.
        \item The monoid operation is the union  $\cup$ of multisets.
        \item The neutral element of the monoid is the empty multiset $\emptyset$.
        \item $X \in D_M$ iff $X = \bigcap\{\mathsf{Pr}_{cf}(C) \mid C \in \mathbb{F}\}$ for some set $\mathbb{F}$ of formulas.
        \item $V(p) := \mathsf{Pr}_{cf}(p)$.
        \item $\mathcal{D}$ is defined as follows:
\[
    \mathcal{D} = \{ X \in D_M \mid [A] \in X \text{ and } X \subseteq \mathsf{Pr}_{cf}(A)
    \text{ for some formula $A \in \mathcal{L}_\mu$}
    \},
\]
where ``$[A]$'' denotes the singleton multiset of $A$.
    \end{itemize}
\end{definition}
\noindent The definition of $\mathcal{D}$ above is similar to that of De et al.~\cite[Definition 33]{DeJafarrahmaniSaurin2022}, who refer to it as the set of \emph{reducibility candidates}. This terminology is inspired by the work of Okada~\cite{Okada2002} and the normalization proofs for various $\lambda$-calculus systems by Tait and Girard~\cite{tait67,girard72}. However, in our definition, $M$, $D_M$ and $\mathcal{D}$ are sets of multisets of not only closed formulas, but also open formulas, i.e., formulas with free variables. This differs from De et al. \cite{DeJafarrahmaniSaurin2022}.

\begin{lemma}
\label{lem:welldefinedSyntacticmu}
The syntactic intuitionistic $\mu$-phase model $(M, D_M, \mathcal{D}, V)$ is well-defined. 
\end{lemma}

\begin{proof}
Well-definedness is established similarly to the case of $\mathbf{IMALL}$ in Section~\ref{sec:Preliminaries}, utilizing Lemma~\ref{inversion} (Inversion). 
To show that $V$ is a $\mathcal{D}$-valuation, it suffices to verify that $\mathsf{Pr}_{cf}(p) \in D_M$ and $[p] \in \mathsf{Pr}_{cf}(p)$ for any $p \in \mathcal{V} \cup \mathcal{A}$. 
The former follows immediately from the definition of $D_M$ by taking $\mathbb{F} = \{p\}$. 
The latter is a consequence of the cut-free provability of the identity sequent $p \vdash p$ in $\mu\mathbf{IMALL}$.
\end{proof}

\subsection{Cut-free Completeness}\label{subsec:CutFreeCompleteness}

The following is a key lemma for establishing the cut-free completeness of $\mu\mathbf{IMALL}^{-}$, generalizing Proposition~\ref{OkadaIMALL} for $\mathbf{IMALL}$.

\begin{lemma}
\label{lem:syntacticMuPhase}
    Let $(M, D_M, \mathcal{D}, V)$ 
    be the syntactic intuitionistic $\mu$-phase model. 
    Let $A$ be a formula, $\vec{B} = B_1, \ldots, B_n$ an $n$-tuple of formulas, $\vec{x} = x_1, \ldots, x_n$ an $n$-tuple of distinct variables, and $\vec{X} = X_1, \ldots, X_n$ an $n$-tuple of closed sets such that $[B_i] \in X_i \subseteq \mathsf{Pr}_{cf}(B_i)$ for all $i$ with $1 \leq i \leq n$. 
Then the following holds: 
\[
[A(\vec{B}/\vec{x})] \in \llbracket A \rrbracket^{V[\vec{x}\mapsto\vec{X}]} \subseteq \mathsf{Pr}_{cf}(A(\vec{B}/\vec{x})),
\] where $V[\vec{x}\mapsto\vec{X}]$ denotes $V[x_1 \mapsto X_1]\cdots[x_n \mapsto X_n]$.
\end{lemma}
\begin{proof}
    By induction on the complexity of $A$, we prove the following: for any $n \in \mathbb{N}$, formulas $B_1, \ldots, B_n \in \mathcal{L}_\mu$, distinct variables $x_1, \ldots, x_n \in \mathcal{V}$, closed sets $X_1, \ldots, X_n \in D_M$, if $[B_i] \in X_i \subseteq \mathsf{Pr}_{cf}(B_i)$ for all $i$ with $1 \leq i \leq n$, then (i) $[A(\vec{B}/\vec{x})] \in \llbracket A \rrbracket^{V[\vec{x}\mapsto\vec{X}]}$ and (ii) $\llbracket A \rrbracket^{V[\vec{x}\mapsto\vec{X}]} \subseteq \mathsf{Pr}_{cf}(A(\vec{B}/\vec{x}))$. We only present the cases where $A$ is a variable, an atomic formula, or a fixpoint formula $A = \mu y.A'$ or $A = \nu y.A'$. The remaining cases follow analogously to the proof of Proposition~\ref{OkadaIMALL} for $\mathbf{IMALL}$ (see~\cite[Lemma 3.6]{Okada2002}). In what follows, ``provable'' refers to cut-free provability in $\mu\mathbf{IMALL}^-$. 
    For the base step, we proceed as follows. 
    \begin{itemize}
        \item Let $A = y \in \mathcal{V}$. 
        %Fix any $n \in \mathbb{N}$, $B_1, \ldots, B_n \in \mathcal{L}_\mu$, $x_1, \ldots, x_n \in \mathcal{V}$, $X_1, \ldots, X_n \in D_M$. 
        Assume that $[B_i] \in X_i \subseteq \mathsf{Pr}_{cf}(B_i)$ for all $1 \leq i \leq n$. If $y = x_i$ for some $i$ such that $1 \leq i \leq n$, then (i) $[y(\vec{B}/\vec{x})] = [B_i] \in X_i = V[\vec{x}\mapsto\vec{X}](x) = \llbracket y \rrbracket^{V[\vec{x}\mapsto\vec{X}]}$ and (ii) $\llbracket y \rrbracket^{V[\vec{x}\mapsto\vec{X}]} = X_i \subseteq \mathsf{Pr}_{cf}(B_i)$ = $\mathsf{Pr}_{cf}(y(\vec{B}/\vec{x}))$ by assumption. 
        Suppose otherwise. 
        For (i), since $y(\vec{B}/\vec{x}) = y$, $[y] \in \mathsf{Pr}_{cf}(y)  = V(y) = V[\vec{x}\mapsto\vec{X}](y) = \llbracket y \rrbracket^{V[\vec{x}\mapsto\vec{X}]}$. 
        For (ii),  we proceed as follows:
        $\llbracket y \rrbracket^{V[\vec{x}\mapsto\vec{X}]}$ = $\mathsf{Pr}_{cf}(y)$ = $\mathsf{Pr}_{cf}(y(\vec{B}/\vec{x}))$, as desired.
        \item Let $A = a \in \mathcal{A}$. This case can be shown similarly to the latter part of the previous one since $a(\vec{B}/\vec{x}) = a$.
    \end{itemize}
For the inductive step, our argument proceeds as follows.
    \begin{itemize}
        \item Let $A = \mu y.A'$.
        Assume that $[B_i] \in X_i \subseteq \mathsf{Pr}_{cf}(B_i)$ for all $1 \leq i \leq n$. 
        Moreover, we can assume without loss of generality that $y \not\in \{x_1, \ldots, x_n\}$ and that $y \not\in \mathsf{FV}(B_i)$ for all $1 \leq i \leq n$. Thus, $(\mu y.A')(\vec{B}/\vec{x})$
        is $\mu y. A'(\vec{B}/\vec{x})$. 
        In what follows, we show that (i) $[\mu y.A'(\vec{B}/\vec{x})] \in \llbracket \mu y.A' \rrbracket^{V[\vec{x}\mapsto\vec{X}]}$, and that (ii) $\llbracket \mu y.A' \rrbracket^{V[\vec{x}\mapsto\vec{X}]} \subseteq \mathsf{Pr}_{cf}(\mu y.A'(\vec{B}/\vec{x}))$.
        \begin{itemize}
            \item  (i) Since $\llbracket \mu y.A' \rrbracket^{V[\vec{x}\mapsto\vec{X}]} := \bigcap \{Y \in \mathcal{D} \mid \llbracket A' \rrbracket^{V[\vec{x}\mapsto\vec{X}][y \mapsto Y]} \subseteq Y\}$, let us fix any $Y \in \mathcal{D}$ such that $\llbracket A' \rrbracket^{V[\vec{x}\mapsto\vec{X}][y \mapsto Y]} \subseteq Y$. 
            Our goal is to show that $[\mu x.A'(\vec{B}/\vec{x})] \in Y$.
            Then it suffices to show that $[\mu y.A'(\vec{B}/\vec{x})] \in \llbracket A' \rrbracket^{V[\vec{x}\mapsto\vec{X}][y \mapsto Y]}$. 
            Since $\llbracket A' \rrbracket^{V[\vec{x}\mapsto\vec{X}][y \mapsto Y]} \in D_M$, it is of the form $\bigcap \{\mathsf{Pr}_{cf}(C) \mid C \in \mathbb{F}\}$ for some $\mathbb{F} \subseteq \mathcal{L}_\mu$. Fix any $C \in \mathbb{F}$. It suffices to show that $[\mu y.A'(\vec{B}/\vec{x})] \in \mathsf{Pr}_{cf}(C)$, i.e., $\mu y.A'(\vec{B}/\vec{x}) \vdash C$ is provable. 
            Since $Y \in \mathcal{D}$, there is some formula $E \in \mathcal{L}_\mu$ such that $[E] \in  Y\subseteq \mathsf{Pr}_{cf}(E)$. By induction hypothesis for (i), we get  $[A'(\vec{B}/\vec{x})(E/y)] \in \llbracket A' \rrbracket^{V[\vec{x}\mapsto\vec{X}][y \mapsto Y]}$. Since $[A'(\vec{B}/\vec{x})(E/y)] \in \bigcap \{\mathsf{Pr}_{cf}(C) \mid C \in \mathbb{F}\}$, a sequent $A'(\vec{B}/\vec{x})(E/y) \vdash C$ is provable. It follows that if $A'(\vec{B}/\vec{x})(E/y) \vdash E$ is provable, then $\mu y.A'(\vec{B}/\vec{x}) \vdash C$ is provable in terms of the functoriality rule (Lemma \ref{Functoriality}) as follows:
            \[
                \begin{bprooftree}
                    \AxiomC{$A'(\vec{B}/\vec{x})(E/y) \vdash E$}
                    \RightLabel{($func$)}
                    \UnaryInfC{$A'(\vec{B}/\vec{x})(A'(\vec{B}/\vec{x})(E/y)/y) \vdash A'(\vec{B}/\vec{x})(E/y)$}
                    \AxiomC{$A'(\vec{B}/\vec{x})(E/y) \vdash C$}
                    \RightLabel{($\mu{\vdash}$)}
                    \BinaryInfC{$\mu y.A'(\vec{B}/\vec{x}) \vdash C.$}
                \end{bprooftree}
            \] 
            We can apply the functoriality rule since we have assumed that the bound variable $y$ does not cause a variable clash in the substitution $(B/x)$, and so we can say that $y$ is still positive in $A'(B/x)$. Finally, we show that $A'(\vec{B}/\vec{x})(E/y) \vdash E$ is provable, i.e., $[A'(\vec{B}/\vec{x})(E/y)] \in \mathsf{Pr}_{cf}(E)$. Since $[A'(\vec{B}/\vec{x})(E/y)] \in \llbracket A' \rrbracket^{V[\vec{x}\mapsto\vec{X}][y \mapsto Y]}$, it suffices to show that $\llbracket A' \rrbracket^{V[\vec{x}\mapsto\vec{X}][y \mapsto Y]} \subseteq \mathsf{Pr}_{cf}(E)$, which holds by $\llbracket A' \rrbracket^{V[\vec{x}\mapsto\vec{X}][y \mapsto Y]} \subseteq Y \subseteq \mathsf{Pr}_{cf}(E)$.
            \item (ii) We show $\llbracket \mu y.A' \rrbracket^{V[\vec{x}\mapsto\vec{X}]} \subseteq \mathsf{Pr}_{cf}(\mu y.A'(\vec{B}/\vec{x}))$. By the definition of $\llbracket \mu x.A'\rrbracket^{V[\vec{x}\mapsto\vec{X}]}$, it suffices to show that $\llbracket A' \rrbracket^{V[\vec{x}\mapsto\vec{X}][y \mapsto Y]} \subseteq Y$ and $Y \subseteq \mathsf{Pr}_{cf}(\mu y.A'(\vec{B}/\vec{x}))$ for some $Y \in \mathcal{D}$. Put $Y = \llbracket A' \rrbracket^{V[\vec{x}\mapsto\vec{X}][y \mapsto \mathsf{Pr}_{cf}(\mu y.A'(\vec{B}/\vec{x}))]}$. For our goal, it suffices to prove that $Y \subseteq \mathsf{Pr}_{cf}(\mu y.A'(\vec{B}/\vec{x}))$. This is because $Y \subseteq \mathsf{Pr}_{cf}(\mu y.A'(\vec{B}/\vec{x}))$ implies that $\llbracket A' \rrbracket^{V[\vec{x}\mapsto\vec{X}][y \mapsto Y]} \subseteq Y$ by the definition of $Y$ and Lemma \ref{Monotonicity} (Monotonicity). So, let us establish $Y \subseteq \mathsf{Pr}_{cf}(\mu y.A'(\vec{B}/\vec{x}))$ below. By the (${\vdash}\mu$) rule, we obtain that $\mathsf{Pr}_{cf}(A'(\vec{B}/\vec{x})(\mu y.A'(\vec{B}/\vec{x})/y)) \subseteq \mathsf{Pr}_{cf}(\mu y.A'(\vec{B}/\vec{x}))$. Hence, it suffices to show $Y \subseteq \mathsf{Pr}_{cf}(A'(\vec{B}/\vec{x})(\mu y.A'(\vec{B}/\vec{x})/y))$. But, this is immediate from induction hypothesis for (ii) and $[\mu y.A'(\vec{B}/\vec{x})] \in \mathsf{Pr}_{cf}(\mu y.A'(\vec{B}/\vec{x}))$. 
        \end{itemize}     

        \item Let $A = \nu y.A'$. Assume that $[B_i] \in X_i \subseteq \mathsf{Pr}_{cf}(B_i)$ for all $1 \leq i \leq n$. Moreover, we can assume without loss of generality that $y \not\in \{x_1, \ldots, x_n\}$ and that $y \not\in \mathsf{FV}(B_i)$ for all $1 \leq i \leq n$. We show (i) $[\nu x.A'(\vec{B}/\vec{x})] \in \llbracket \nu y.A' \rrbracket^{V[\vec{x}\mapsto\vec{X}]}$, and (ii) $\llbracket \nu y.A' \rrbracket^{V[\vec{x}\mapsto\vec{X}]} \subseteq \mathsf{Pr}_{cf}(\nu y.A'(\vec{B}/\vec{x}))$. 
        \begin{itemize}
            \item  (i) Recall that $\llbracket \nu y.A' \rrbracket^{V[\vec{x}\mapsto\vec{X}]} \coloneqq cl(\bigcup \{Y \in \mathcal{D} \mid Y \subseteq \llbracket A' \rrbracket^{V[\vec{x}\mapsto\vec{X}][y \mapsto Y]}\})$. It suffices to show that $[\nu y.A'(\vec{B}/\vec{x})] \in \bigcup \{Y \in \mathcal{D} \mid Y \subseteq \llbracket A' \rrbracket^{V[\vec{x}\mapsto\vec{X}][y \mapsto Y]}\}$ by Proposition \ref{closureProposition}. For our goal, it suffices to prove that $[\nu y.A'(\vec{B}/\vec{x})] \in Y$ and $Y \subseteq \llbracket A' \rrbracket^{V[\vec{x}\mapsto\vec{X}][y \mapsto Y]}$ for some $Y \in \mathcal{D}$. Put $Y = \bigcap \{\mathsf{Pr}_{cf}(C) \mid \text{$C \in \mathcal{L}_\mu$ and $[\nu y.A'(\vec{B}/\vec{x})] \in \mathsf{Pr}_{cf}(C)$}\}$. It is clear that $Y \in D_M$. Moreover, we have $[\nu y.A'(\vec{B}/\vec{x})] \in Y$ and $Y \subseteq \mathsf{Pr}_{cf}(\nu y.A'(\vec{B}/\vec{x}))$, hence $Y \in \mathcal{D}$. So, it suffices to show that $Y \subseteq \llbracket A' \rrbracket^{V[\vec{x}\mapsto\vec{X}][y \mapsto Y]}$ in what follows. Since $\llbracket A' \rrbracket^{V[\vec{x}\mapsto\vec{X}][y \mapsto Y]} \in D_M$, it is of the form $\bigcap \{\mathsf{Pr}_{cf}(C) \mid C \in \mathbb{F}\}$ for some $\mathbb{F} \subseteq \mathcal{L}_\mu$. It suffices to show that $[\nu y.A'(\vec{B}/\vec{x})] \in \bigcap \{\mathsf{Pr}_{cf}(C) \mid C \in \mathbb{F}\}$ by $[\nu y.A'(\vec{B}/\vec{x})] \in Y$ and the definition of $Y$. 
            Fix any $C \in \mathbb{F}$. We show that $[\nu y.A'(\vec{B}/\vec{x})] \in \mathsf{Pr}_{cf}(C)$, i.e., $\nu y.A'(\vec{B}/\vec{x}) \vdash C$ is provable. By induction hypothesis for (i), we get 
            $$[A'(\vec{B}/\vec{x})(\nu y.A'(\vec{B}/\vec{x})/y)] \in \llbracket A' \rrbracket^{V[\vec{x}\mapsto\vec{X}][y \mapsto Y]} = \bigcap \{\mathsf{Pr}_{cf}(C) \mid C \in \mathbb{F}\}$$ 
            Therefore, $[A'(\vec{B}/\vec{x})(\nu y.A'(\vec{B}/\vec{x})/y)] \in \mathsf{Pr}_{cf}(C)$, i.e., $A'(\vec{B}/\vec{x})(\nu y.A'(\vec{B}/\vec{x})/y) \vdash C$ is provable. By applying the ($\nu{\vdash}$) rule to this sequent, we obtain a proof of $\nu y.A'(\vec{B}/\vec{x}) \vdash C$.

        \item  (ii) We show that $cl(\bigcup \{Y \in \mathcal{D} \mid Y \subseteq \llbracket A' \rrbracket^{V[\vec{x}\mapsto\vec{X}][y \mapsto Y]}\}) \subseteq \mathsf{Pr}_{cf}(\nu y.A'(\vec{B}/\vec{x}))$. 
        
        Since $\mathsf{Pr}_{cf}(\nu y.A'(\vec{B}/\vec{x})) \in D_M$, it suffices to show that $\bigcup \{Y \in \mathcal{D} \mid Y \subseteq \llbracket A' \rrbracket^{V[\vec{x}\mapsto\vec{X}][y \mapsto Y]}\} \subseteq \mathsf{Pr}_{cf}(\nu y.A'(\vec{B}/\vec{x}))$ by Proposition \ref{closureProposition}. Fix any $Y \in \mathcal{D}$ such that $Y \subseteq \llbracket A' \rrbracket^{V[\vec{x}\mapsto\vec{X}][y \mapsto Y]}$. Our goal is to show that $Y \subseteq \mathsf{Pr}_{cf}(\nu y.A'(\vec{B}/\vec{x}))$. Then it suffices to show that $\llbracket A' \rrbracket^{V[\vec{x}\mapsto\vec{X}][y \mapsto Y]} \subseteq \mathsf{Pr}_{cf}(\nu y.A'(\vec{B}/\vec{x}))$. Since $Y \in \mathcal{D}$, there is some $E \in \mathcal{L}_\mu$ such that $[E] \in Y \subseteq \mathsf{Pr}_{cf}(E)$. By induction hypothesis for (ii), we get $\llbracket A' \rrbracket^{V[\vec{x}\mapsto\vec{X}][y \mapsto Y]} \subseteq \mathsf{Pr}_{cf}(A'(\vec{B}/\vec{x})(E/y))$. Hence, it suffices to show $\mathsf{Pr}_{cf}(A'(\vec{B}/\vec{x})(E/y)) \subseteq \mathsf{Pr}_{cf}(\nu y.A'(\vec{B}/\vec{x}))$. Fix any $\Gamma \in \mathsf{Pr}_{cf}(A'(\vec{B}/\vec{x})(E/y))$. We show that $\Gamma \in \mathsf{Pr}_{cf}(\nu y.A'(\vec{B}/\vec{x}))$, i.e., $\Gamma \vdash \nu y.A'(\vec{B}/\vec{x})$ is provable. It follows that 
            if $E \vdash A'(\vec{B}/\vec{x})(E/y)$ is provable then $\Gamma \vdash \nu y.A'(\vec{B}/\vec{x})$ is provable in terms of the functoriality rule (Lemma \ref{Functoriality}) as follows:
        \[
            \begin{bprooftree}
                \AxiomC{$E \vdash A'(\vec{B}/\vec{x})(E/y)$}
                \RightLabel{($func$)}
                \UnaryInfC{$A'(\vec{B}/\vec{x})(E/y) \vdash A'(\vec{B}/\vec{x})(A'(\vec{B}/\vec{x})(E/y)/y)$}
                \AxiomC{$\Gamma \vdash A'(\vec{B}/\vec{x})(E/y)$}
                \RightLabel{(${\vdash}\nu$)}
                \BinaryInfC{$\Gamma \vdash \nu y.A'(\vec{B}/\vec{x})$.}
            \end{bprooftree}
        \]
        Finally, we show that $E \vdash A'(\vec{B}/\vec{x})(E/y)$ is provable, i.e., $[E] \in \mathsf{Pr}_{cf}(A'(\vec{B}/\vec{x})(E/y))$. This holds by $[E] \in Y \subseteq \llbracket A' \rrbracket^{V[\vec{x}\mapsto\vec{X}][y \mapsto Y]} \subseteq \mathsf{Pr}_{cf}(A'(\vec{B}/\vec{x})(E/y))$.
        This finishes establishing 
        $\Gamma \in \mathsf{Pr}_{cf}(\nu y.A'(\vec{B}/\vec{x}))$. 
        \qedhere
        \end{itemize}
    \end{itemize}
\end{proof}
\begin{remark}
    Our proof strategy differs from that of De et al.~\cite[Lemmas 34 and 35]{DeJafarrahmaniSaurin2022}. 
    In their approach, they first established (ii) for all formulas and subsequently derived (i) as a consequence of (ii). In the classical setting, Negation Normal Form (NNF) allows a simpler induction. In contrast, in the intuitionistic case, NNF is unavailable, requiring a two-sided calculus to deal with the rules of linear implication $\multimap$. Then, our proof requires proving (i) and (ii) simultaneously for each case, because the inductive step for linear implication relies on the inductive hypotheses of both (i) and (ii) (recall the proof of Proposition \ref{OkadaIMALL} (see also~\cite[Lemma 3.6]{Okada2002})). This simultaneous induction carries over to the fixpoint cases.
\end{remark}

\begin{lemma}\label{lem:admissibilityOfSyntacticModel}
For all formulas $A \in \mathcal{L}_\mu$, it holds that $[A] \in \llbracket A \rrbracket^V \subseteq \mathsf{Pr}_{cf}(A)$.
Therefore, the syntactic intuitionistic $\mu$-phase model of Definition \ref{dfn:syntacticMuPhaseModel} is admissible. 
\end{lemma}

\begin{proof}
Fix any formula $A \in \mathcal{L}_\mu$. 
By Lemma \ref{lem:syntacticMuPhase}, we get $[A(\vec{x}/\vec{x})] \in \llbracket A \rrbracket^{V[\vec{x}\mapsto \vec{V(x)}]} \subseteq \mathsf{Pr}_{cf}(A(\vec{x}/\vec{x}))$ hence 
$[A] \in \llbracket A \rrbracket^V \subseteq \mathsf{Pr}_{cf}(A)$. 
\end{proof}

Moreover, we can prove the cut-free completeness and the cut-elimination theorem for $\mu\mathbf{IMALL}$ similarly to $\mathbf{IMALL}$.

\begin{lemma}[Cut-free Completeness for $\mu\mathbf{IMALL}$]\label{lem:CutFreeCompletenessForMuIMALL}
    If a formula $A$ is true in any admissible intuitionistic $\mu$-phase model, then $\vdash A$ is provable in $\mu\mathbf{IMALL}^-$.
\end{lemma}
\begin{proof}
    Assume that $A$ is true in any admissible intuitionistic $\mu$-phase model. Then, by Lemma \ref{lem:admissibilityOfSyntacticModel}, $A$ is true in the syntactic intuitionistic $\mu$-phase model $(M, D_M, V)$, meaning $1 = \emptyset \in \llbracket A \rrbracket^V$. By Lemma \ref{lem:syntacticMuPhase}, we get  $\llbracket A \rrbracket^V \subseteq \mathsf{Pr}_{cf}(A)$. Therefore, we conclude that $\emptyset \in \mathsf{Pr}_{cf}(A)$, which implies that $\vdash A$ is provable in $\mu\mathbf{IMALL}^-$.
\end{proof}
\begin{theorem}[Cut-elimination for $\mu\mathbf{IMALL}$]
    If a sequent $\Gamma \vdash C$ is provable in $\mu\mathbf{IMALL}$, then it is provable in $\mu\mathbf{IMALL}^-$. 
\end{theorem}
\begin{proof}
    Assume that $\Gamma \vdash C$ is provable in $\mu\mathbf{IMALL}$. By applying the ($\otimes{\vdash}$) rules several times and the ($\vdash\multimap$) rule once, the sequent $\vdash \bigotimes \Gamma \multimap C$ is provable in $\mu\mathbf{IMALL}$. By Lemma \ref{SoundnessForIMALL} (Soundness for $\mu\mathbf{IMALL}$), the formula $\bigotimes \Gamma \multimap C$ is true in any admissible intuitionistic $\mu$-phase model. By Lemma \ref{lem:CutFreeCompletenessForMuIMALL} (Cut-free Completeness for $\mu\mathbf{IMALL}$), the sequent $\vdash \bigotimes \Gamma \multimap C$ is provable in $\mu\mathbf{IMALL}^-$. Then, by Lemma \ref{inversion} (Inversion in $\mu\mathbf{IMALL}^-$), $\Gamma \vdash C$ is provable in $\mu\mathbf{IMALL}^-$.
\end{proof}
\section{Conclusion and Future Directions}\label{sec:ConclusionAndFutureWorks}
In this paper, we have defined phase semantics for intuitionistic propositional multiplicative-additive linear logic with least and greatest fixpoints, $\mu \mathbf{IMALL}$, and established the cut-elimination theorem for the system via these semantics. 

There are six potential directions for future research. Firstly, we propose to extend our phase semantics and cut-elimination proof to a first-order predicate system. Since linear logic with least and greatest fixpoints was originally introduced as a first-order system within the context of linear logic programming, this extension is expected to be highly fruitful.

Secondly, we aim to prove the cut-elimination theorem using the two alternative methods outlined in the introduction. Specifically, it may be possible to establish cut-elimination by defining reduction rules, potentially by translating formulas into second-order linear logic.

Thirdly, it is necessary to investigate the relationship between the exponential modality $\oc$ and the fixpoint operators within an intuitionistic setting. Baelde~\cite[Section 2.3]{Baelde2012} demonstrated that the exponential modality $\oc$ can be simulated in classical linear logic via the translation $t(\oc A) \coloneqq \nu x. (\mathbf{1} \with A \with (x \otimes x))$. However, Baelde noted the difficulties inherent in the converse direction of this simulation, and Das~\cite{Das2024} subsequently established that the translation is not faithful. We aim to investigate whether this lack of faithfulness also holds in intuitionistic linear logic. As suggested by Das~\cite{Das2024}, constructing counterexamples using phase semantics is a promising approach for proving non-faithfulness, which aligns closely with the research presented in this paper.

Fourthly, conservativity results are of interest. Schellinx~\cite[Proposition 3.8]{Schellinx1991} showed syntactically that $\mathbf{CLL}$ is a conservative extension of $\mathbf{ILL}$ without $\mathbf{0}$. However, this relationship does not hold in the presence of fixpoint operators, as $\mathbf{0}$ is equivalent to $\mu x.x$. It remains to be investigated which syntactic fragments containing fixpoint operators preserve this conservativity. Furthermore, Schellinx's result can be viewed as an embedding of $\mathbf{ILL}$ into $\mathbf{CLL}$ without $\mathbf{0}$. In contrast, Kanovich, Okada, and Terui~\cite{KanovichOkadaTerui2006} established an embedding of $\mathbf{ILL}$ into $\mathbf{CLL}$ by studying the relationship between classical and intuitionistic phase spaces. It would also be interesting to investigate whether their semantic results could be adapted to $\mu\mathbf{MALL}$ and $\mu\mathbf{IMALL}$ to clarify the relationship between these systems.

Fifthly, we could apply our phase semantics and cut-elimination proof to a system incorporating $\omega$-rules. De et al.~\cite[Section 4]{DeJafarrahmaniSaurin2022} introduced the system $\mu \mathbf{MALL}_\omega$, in which the greatest fixpoint $\nu$ is handled by an $\omega$-rule with infinitely many premises. This system incorporates induction and coinduction implicitly, and as a result, it enjoys a form of the subformula property. In contrast, the system presented in this paper lacks this property due to the invariant $S$ in the $(\mu{\vdash})$ and $({\vdash}\nu)$ rules. De et al. also defined phase semantics for $\mu\mathbf{MALL}_\omega$ and established the cut-elimination theorem using these semantics. Furthermore, it may be possible to define an intuitionistic fragment of $\mu\mathbf{MALL}_\omega$, denoted as $\mu\mathbf{IMALL}_\omega$, and develop its phase semantics to prove cut-elimination. Given that Terui \cite{Terui2018} established cut-elimination for second-order intuitionistic logic with similar $\Omega$-rules using algebraic semantics, this approach appears feasible.

Finally, this work paves the way for the development of higher-order model checking using $\mu \mathbf{IMALL}$ as a specification logic. Our long-term objective is to derive decidability arguments formulated directly in terms of inductive and coinductive principles, with the ultimate goal of certifying these proofs within proof assistants. This direction contributes to the construction of trustworthy formal method tools for functional programming, facilitating the transfer of highly theoretical results into practical and reliable applications.

\section*{Acknowledgements}
The work of the first author was supported by JST SPRING, Grant Number JPMJSP2119. The work of the third author was partially supported by JSPS KAKENHI Grants-in-Aid for Scientific Research (B) (Grant Number JP22H00597) and (C) (Grant Number JP25K03537).
\nocite{*}
\bibliographystyle{eptcs}
\bibliography{generic}
\end{document}